\documentclass[12pt]{article}

\newcommand{\blind}{1}
\usepackage[margin=1in]{geometry} 

\usepackage[letterpaper=true,colorlinks=true,pdfpagemode=none,urlcolor=blue,linkcolor=blue,citecolor=blue,pdfstartview=FitH]{hyperref}
\usepackage{amsmath,amsfonts,amssymb}
\usepackage{graphicx}
\usepackage{color}
\usepackage[round,sort&compress]{natbib}

\usepackage{plsMacros}

\usepackage{plsThms}
\renewcommand{\hat}{\widehat}

\newtheorem{remark}{Remark}

\usepackage{mathtools}
\mathtoolsset{showonlyrefs,showmanualtags}

\begin{document}

\if1\blind
{
\title{Compressed and Penalized Linear Regression}
\author{Darren Homrighausen\thanks{
    Darren Homrighausen is currently Visiting Assistant Professor at
    Southern Methodist University. The authors gratefully acknowledge
    support from the National Science Foundation (grant DMS-1407543 to
  DH; grant DMS-1407439 to DJM) and the Institute for New Economic Thinking
  (grant INO14–00020 to DH and DJM).}\hspace{.2cm}\\ 
    Department of Statistics, Colorado State University\\
  and\\
  Daniel J. McDonald\\
  Department of Statistics, Indiana University}
\maketitle
} \fi

\if0\blind
{
  \bigskip
  \bigskip
  \bigskip
  \begin{center}
    {\LARGE\bf Compressed and Penalized Linear Regression}
\end{center}
  \medskip
} \fi

\bigskip
\begin{abstract}
Modern applications require methods that are computationally
feasible on large datasets but also preserve
statistical efficiency. Frequently, these two concerns are seen as
contradictory: approximation methods that enable computation are
assumed to degrade statistical performance relative to exact
methods. In applied mathematics, where
much of the current theoretical work on approximation resides, the inputs
are considered to be observed exactly.  
The prevailing philosophy is that while the exact problem is, regrettably, unsolvable,
any approximation should be as small as possible. 
However, from a statistical perspective, an approximate or 
regularized solution may be preferable to the exact one. 
Regularization formalizes a trade-off between
fidelity to the data and adherence to prior knowledge about the data-generating
process such as smoothness or sparsity. The resulting estimator
tends to be more useful, interpretable, and suitable as an input to
other methods.

In this paper, we propose new methodology for estimation and prediction 
under a linear model borrowing insights
from the approximation literature.
We explore these procedures from a
statistical perspective and find that in many cases they
improve both computational and statistical performance.

\end{abstract}

\noindent%
\vfill




\section{Introduction}

Recent work in large-scale data analysis has focused on 
developing fast and randomized
approximations to important numerical linear algebra tasks 
such as solving least squares
problems \citep{Woodruff2014,RokhlinTygert2008,drineas2011faster} and
finding spectral decompositions \citep{halko2011finding,Gittens:2013aa,HomrighausenMcDonald2016}. 
These approaches, known as {\it compression}, {\it sketching}, or {\it
  preconditioning}, take a given data set and construct a reduced sized,
``compressed,'' solution.  
Theoretical justifications for these approximate approaches 
upper bounds the difference in the objective function 
evaluated at the ``compressed'' solution relative to the full-data solution.

In numerical linear algebra, a major goal is to develop compression algorithms that 
decrease the computational or storage burden 
while not giving up too much accuracy.
Due to their random nature,
theoretical performance bounds are stated with high probability with respect to the subsampling or random projection mechanism
\citep{halko2011finding}.  For example, when confronted with a massive least squares problem, 
we could attempt to form a compressed solution such that its residual sum of squares is not too much larger than
the residual sum of squares of the least squares solution.

In this paper, we take a more statistical perspective: we investigate the performance in terms of
 parameter estimation and prediction risk.  Leveraging insights into the
 statistical behavior of the most commonly used compressions, 
 we develop and explore novel algorithms for compressed least squares.
 We show that our methods
provide satisfactory, or even {\it superior}, statistical performance at a fraction of the
computation and storage costs.
%
%
%
%

\subsection{Overview of the problem}
\label{sec:overview-problem}
 Suppose 
we make $n$ paired, independent observations $X_i \in \R^p$ and  $Y_i\in \R$, $i = 1,\ldots,n$, where $X_i$ is a vector
of measurements, $Y_i$ is the associated response, $n \gg p$, and both $n$ and $p$ are both very large.
Concatenating the vectors $X_i$ row-wise into a matrix $\X \in \R^{n\times p}$ 
and the responses $Y_i$ into a vector $\Y$, we assume that 
there exists a $\b_* \in \R^p$ such that 
\begin{equation}
\Y = \X \b_* + \sigma\varepsilon,
\label{eq:linearModel}
\end{equation} 
with $\E \varepsilon = 0$, 
$\V \varepsilon = I_n$, and $\X$ fixed. 

We can seek to estimate a relationship between
$X$ and $Y$ via linear regression.  That is, for any vector $\b\in \R^p$, we define 
$\ell(\b) := \sum_{i=1}^n (Y_i - X_i^{\top}\b)^2$.
Writing the (squared) Euclidean norm as
$\norm{\X\b-\Y}_2^2 := \sum_{i=1}^n (Y_i - X_i^{\top}\b)^2$, a least squares solution
is  a vector $\hat{\b}\in \R^p$ such that
\begin{equation}
\label{eq:leastSquaresSolution}
\ell(\hat{\b}) = \min_\b \norm{\X\b-\Y}_2^2.
\end{equation}  
This is given by $\hat{\b} = \X^\dagger \Y$, where $\X^\dagger$ is the
Moore-Penrose pseudo inverse of $\X$. If $\X$ has full column rank,
the solution simplifies to 
$\hat{\b} = (\X^{\top}\X)^{-1} \X^{\top}\Y$.

The least squares solution can be computed stably in $O(np^2)$ time using, for example,
routines in LAPACK such as the QR decomposition, Cholesky decomposition of the normal equations, 
or the singular value decomposition \citep{golub2012matrix}.  It also
has a few, well-known statistical 
properties such as being a minimum variance unbiased estimator or the
best linear unbiased estimator. 
However, classic numerical linear algebraic techniques require substantial 
random access to $\X$ and $\Y$, so computing $\hat{\b}$ can be infeasible or undesirable
in practice.

The big-data regime we consider, i.e. $n \gg p$ and both are very large, can happen in many different 
scientific areas such as psychology, where 
cellular phones are used to collect high-frequency data on
individual actions; atmospheric science, where
multiresolution satellite images are used to understand climate change and predict
future weather patterns; technology companies, which use
massive customer databases to predict tastes and preferences; or astronomy,
where hundreds of millions of objects are measured using radio telescopes.

%
%
%
%

\subsection{Prior work}
\label{sec:prior-results}


A very popular approach in the approximation literature
\citep{RokhlinTygert2008,drineas2011faster,Woodruff2014} is to
generalize equation~\eqref{eq:leastSquaresSolution} to
include a compression matrix $Q \in \R^{q \times n}$, with $n \geq q > p$, 
$\ell_Q(\b) =  \norm{Q\X\b-Q\Y}_2^2$.
The associated approximation
to $\hat{\b}$ is the {\it fully compressed estimator}
\begin{equation}
\hat{\b}_{FC} = \argmin_{\beta} \norm{Q\X\b-Q\Y}_2^2.
\label{eq:fullCompression}
\end{equation}  
Now, $\hat{\b}_{FC}$ can be computed via standard techniques by using the
compressed data $Q\X$ and $Q\Y$.
We defer discussion of strategies and trade-offs for specific choices of $Q$
to \autoref{sec:Q}, but, clearly, some structure on $Q$ is required to enable fast
multiplication.

The standard theoretical justification defines a tolerance parameter $\epsilon$ and a compression 
parameter $q = q(\epsilon)$ such that with high probability \citep{drineas2011faster,drineas2012fast},
\begin{equation}
\ell(\hat\beta_{FC})  \leq (1 + \epsilon) \ell(\hat\beta).
\label{eq:genApproxResult}
\end{equation}
In this case, the probability is stated with respect to the process that generates  $Q$ only
and the data are considered fixed.  Thus~\eqref{eq:genApproxResult}
is a worst-case analysis since it must hold uniformly over all data sets, regardless
of the ``true'' data generating process. 

There have been many proposals
for how to choose the compression matrix $Q$ and the compression parameter $q$.  
A classical approach is to define $Q$ such that $Q\X$
corresponds to uniform random sampling of the rows of $\X$. 
The success of this approach, in the sense of obtaining small $\epsilon$ and $q$ in \eqref{eq:genApproxResult},
depends crucially on the coherence of $\X$, defined as the maximum
squared-Euclidean norm of  
a row of an orthogonal matrix that spans the column space of $\X$
\citep{AvronMaymounkov2010,drineas2011faster,RokhlinTygert2008}.
Note that  the coherence is very different from the condition number, which is the ratio
of the largest and smallest singular values of $\X$: the coherence is necessarily
in the interval $[p/n,1]$ while the condition number can be any positive real number.  If $\X$ has
coherence $p/n$, then $q = O(p\log p)$ is sufficient to obtain a high probability bound while
a coherence of $1$ essentially requires all rows to be sampled \citep{AvronMaymounkov2010}.

Therefore, it is important to either reduce the coherence or sample
the rows in proportion to their influence on it \citep[this
  influence is known as a ``leverage score''][]{drineas2006sampling}. 
Unfortunately, it is as expensive to compute the 
coherence or the leverage scores as it is to solve the original least squares problem in \eqref{eq:leastSquaresSolution}.
Hence, if computational complexity is of concern, they are both unavailable.

Instead, a general approach for reducing the coherence of $\X$ is to randomly 
``mix'' its rows. The rows of this new, mixed matrix have smaller
coherence (with high probability) and can then be randomly sampled.
In particular, \citet{drineas2011faster} show that equation \eqref{eq:genApproxResult} holds
as long as $q$ grows like $p\log(n) \log(p\log(n))$; though it is claimed that smaller $q$ still works well in practice.
Readers interested in the practical performance of 
these randomized algorithms should see
Blendenpik~\citep{AvronMaymounkov2010} or LSRN~\citep{MengSaunders2014}.

Relative to the computational properties of these compression methods,
there has been comparatively little work on their statistical properties.
\citet{raskutti2015statistical} analyze various
relative efficiency measures of $\hat\beta_{FC}$ versus $\hat{\beta}$ as a function of the compression
matrix $Q$.  They find that the statistical quality of $\hat\beta_{FC}$ depends on the oblique projection
matrix $U(QU)^{\dagger}U$, where $U$ is the left singular
matrix of $\X$.
Additionally, \citet{MaMahoney2015} develop a
theoretical framework, which we adopt, to compare the statistical
performance of various compression matrices used
in~\eqref{eq:fullCompression}. In particular, they show that
in terms of the mean squared error (MSE) of $\hat\b_{FC}$, neither leverage-based sampling nor
uniform sampling uniformly dominates the other. That is, leverage-based sampling
works better for some datasets while uniform sampling has better MSE for others. Despite using their framework, we take a different
perspective: that by combining compression with regularization, we can
improve over the uncompressed solution and produce an estimator with both better computational and
statistical properties.


\subsection{Our contribution}


This paper, in contrast with previous research,
adopts the perspective that approximations mimicking the
least squares estimator $\hat\b$ may produce faster methods, but may
not result in good estimators.
Instead, we seek approximations that minimize
estimation and/or prediction error. 
It is well known that the least squares estimator $\hat\b$, while
being a minimum variance, unbiased estimator, performs 
poorly in terms of prediction or estimation risk relative to
regularized estimators.  As approximation and regularization are 
very similar, this insight suggests the intriguing possibility that it
is possible to develop a compressed estimator 
that performs better statistically, and is cheaper to compute and
store, than $\hat\b$. 

We show that $\hat\beta_{FC}$, like $\hat\beta$, is unbiased and hence
must have a larger estimation and  
prediction error.  As a remedy, we define a {\it partial compression}
estimator, notated  $\hat\beta_{PC}$ 
and defined in equation \eqref{eq:wPC}.  In contrast to
$\hat\beta_{FC}$, $\hat\beta_{PC}$  
is a biased estimator.  In fact, its bias is such that it performs relatively poorly in practice.  
Therefore, we propose a {\it linear combination} of
$\hat\beta_{FC}$ and $\hat\beta_{PC}$.  We find that this
combined estimator performs  
much better than either  individually.


Furthermore, as regularized regression outperforms $\hat\beta$, it is sensible to 
extend $\hat\beta_{FC}$, $\hat\beta_{PC}$, and our linear combination to
analogous regularized versions. 
This introduces a tuning parameter with which to directly calibrate bias and variance.
In this paper, we use  {\it ridge regression} as a regularized least squares method,
though other methods such as bridge, lasso, or the nonnegative garrote are alternatives.
The ridge regression estimator works
by inflating the smallest singular values of $\X$ thereby stabilizing
the least squares problem. 
However, solving for the ridge regression estimator has the same
computational complexity as solving equation
\eqref{eq:linearModel}. 
Therefore, we apply the compression paradigm
to ridge regression.  We find that in many cases, these 
estimators can be computed for a fraction of the cost
while providing better performance than the original least-squares estimator.


In  \autoref{sec:compr-regr} we carefully define the our proposed
estimators. Because these estimators contain a
tuning parameter, \autoref{sec:tuning-param-select} gives a
data-driven procedure for selecting it with minimal extra computation.
\autoref{sec:simulations} examines the performance of these
compressed estimators via simulation. In particular,
we demonstrate that in most cases our methods have better performance
than the least-squares solution and occasionally better performance
than the uncompressed ridge regression estimator. Likewise, \autoref{sec:data-example}
looks at the effectiveness of our estimators on two real-data examples. In
\autoref{sec:theory-results}, we give theoretical expressions for the
bias and variance of our estimators and compare these to the standard
results for ridge regression. We show that, to a first
order approximation, the compressed estimators require a different
tuning parameter than ridge regression to minimize bias and
variance. Lastly, \autoref{sec:conclusion} summarizes our conclusions
and presents avenues for further research.

\section{Compressed regression}
\label{sec:compr-regr}

In this section, we describe the standard set up of compressed least
squares regression before introducing our modifications and the
specific form of the compression matrix $Q$ we consider.

\subsection{Compressed least squares regression}
\label{sec:compr-least-squar}
The fully compressed least squares estimator, defined in \eqref{eq:fullCompression}, can be
written as
\begin{equation}
\min_\b \norm{Q(\X\b - Y)}_2^2 = \min_\b \left( \b^{\top}\X^{\top}Q^{\top}Q\X\b - 2\b^{\top}\X^{\top}Q^{\top}Q\Y\right).
\label{eq:wFC}
\end{equation}
An alternative approach is {\it partial compression} \citep{BeckerKawas2015}
\begin{equation}
\min_\b \left( \b^{\top}\X^{\top}Q^{\top}Q\X\b - 2\b^{\top}\X^{\top} \Y\right),
\label{eq:wPC}
\end{equation}
which removes the compression matrix from the cross-product term.
Depending on the particular draw and form of the compression matrix $Q$, there may not be unique solutions
to equations \eqref{eq:wFC} or \eqref{eq:wPC}.  

\subsection{Compressed ridge regression}

A well used technique to stabilize the least squares problem is known as Tikhonov regularization or ridge regression
in applied mathematics~\citep{TikhonovArsenin1979} and
statistics~\citep{hoerl1970ridge}, respectively.  The ridge regression
problem can  
be written in the Lagrangian form as
\begin{equation}
\hat{\b}(\lambda) := \argmin_\b \norm{\X\b - \Y}_2^2  + \lambda \norm{\b}_2^2.
\label{eq:ridge}
\end{equation}
While $\hat{\b}(\lambda)$ has better numerical
properties than $\hat{\b}(0) \equiv \hat{\b}$ because it inflates the small singular values of $\X$, it
improves neither the computational complexity, which is still
$O(np^2)$, nor the storage, which is $O(np)$.

In addition to better numerical stability, the ridge solution has lower
MSE than $\hat{\b}$ for some $\lambda$.
These results beg the question: if $\hat{\b}(\lambda)$ is a better
overall procedure, why not compress it instead?  Leveraging
this insight, we define the {\it fully compressed ridge estimator}, in
analogue to equation \eqref{eq:wFC}, as
\begin{equation}
\min_\b \norm{Q(\X\b - \Y)}_2^2 + \lambda\norm{\b}_2^2= \min_\b \left(
  \b^{\top}\X^{\top}Q^{\top}Q\X\b -
  2\b^{\top}\X^{\top}Q^{\top}Q\Y\right) + \lambda \b^{\top}\b. 
\label{eq:wFCridge}
\end{equation}
The minimizer to equation \eqref{eq:wFCridge} can be written
\begin{equation}
\hat{\b}_{FC}(\lambda) = (\X^{\top}Q^{\top}Q\X + \lambda
I)^{-1}\X^{\top}Q^{\top}Q\Y. 
\label{eq:wFCridgeEst}
\end{equation}
Likewise, analogous to equation \eqref{eq:wPC}, the {\it partially compressed ridge estimator} solves
\begin{equation}
\min_\b \left( \b^{\top}\X^{\top}Q^{\top}Q\X\b - 2\b^{\top}\X^{\top} \Y\right)  + \lambda \b^{\top}\b
\label{eq:wPCridge}
\end{equation}
with minimizer
\begin{equation}
\hat{\b}_{PC}(\lambda) = (\X^{\top}Q^{\top}Q\X + \lambda I)^{-1}\X^{\top}\Y.
\label{eq:wPCridgeEst}
\end{equation}
Ignoring numerical issues for very small $\lambda$, both of these
estimators always have a unique solution regardless of $Q$ and $\X$.

\subsection{Linear combination compressed ridge regression}
Instead of choosing between $\hat{\b}_{FC}(\lambda)$ and $\hat{\b}_{PC}(\lambda)$, it is reasonable
to use a model averaged estimator formed by combining them.  Consider the estimator generated by a convex combination
\begin{equation}
\hat{\b}_\alpha(\lambda) = \alpha \hat{\b}_{FC}(\lambda) + (1-\alpha) \hat{\b}_{PC}(\lambda),
\label{eq:mostGeneral}
\end{equation}
where $\alpha \in [0,1]$.  
A data-driven value $\hat\alpha$ can be computed
by
forming the matrix 
\[
B(\lambda) =
\begin{bmatrix}
\hat{\beta}_{FC}(\lambda),\, \hat{\beta}_{PC}(\lambda)
\end{bmatrix}
\in
\R^{p\times 2},
\]
a column-wise concatenation of $\hat{\b}_{FC}(\lambda)$ and
$\hat{\b}_{PC}(\lambda)$, and then solving
\begin{equation}
\label{eq:convexCombinationEstimator}
\hat\alpha = \argmin_{\alpha \in [0,1]} \norm{\X B(\lambda)
\begin{bmatrix}
\alpha \\
(1-\alpha)
\end{bmatrix} - Y}_2^2.
\end{equation}
There is no reason that the convex constraint will provide the best
estimator.  Hence, we also consider the
unconstrained two-dimensional least squares problem given by
\begin{equation}
\hat{\alpha}= \argmin_{\alpha  \in \R^2} \norm{\X B(\lambda) \alpha - Y}_2^2.
\label{eq:alphaHat}
\end{equation}
We emphasize that either version can be computed in $O(n)$ time. Lastly, an estimator of $\beta$, or a prediction $\hat{Y}$, can be
produced with
\begin{equation}
\hat{\beta}_{\hat\alpha}(\lambda) := B(\lambda) \hat{\alpha}
\label{eq:linearCombinationEstimator}
\end{equation} 
and 
\begin{equation}
\hat{Y}_{\hat\alpha}(\lambda) := \X \hat{\beta}_{\hat\alpha}(\lambda),
\label{eq:linearCombinationPredictions}
\end{equation} 
respectively.


\subsection{Compression matrices}
\label{sec:Q}
The effectiveness of these approaches depends on $q$, the nature of $Q$, 
and the structure of $\X$ and $\bstar$.  For arbitrary $Q$, the multiplication $Q\X$ would take $O(qnp)$
operations and, hence, could be expensive relative to solving the
original least squares problem. However, this multiplication is
``embarrassingly parallel'' (say, by the map-reduce framework)
rendering the multiplication cost
somewhat meaningless in contrast to
the least squares solution, which is not easily parallelized. 
Therefore, the limiting computation for solving
\eqref{eq:wFCridge} is only $O(qp^2)$.

The structure of $Q$ is chosen, typically, either for its theoretical
or computational properties.  Examples are standard Gaussian entries, producing
dense but theoretically convenient $Q$, fast Johnson-Lindenstrauss  methods, or
the counting sketch.  A thorough discussion of these methods is outside the scope of this paper.  Instead,
we use a ``sparse Bernoulli'' matrix~\citep{KaneNelson2014,DasguptaKumar2010,Woodruff2014,Achlioptas2003}.  Here, the entries
of $Q$ are generated independently where $\P(Q_{ij} = 0) = 1 - 1/s$
and $\P(Q_{ij} = -1) = \P(Q_{ij} = 1) = 1/(2s)$ for some $s\geq 1$.  Then, $\E Q$  has $qn/s$
non-zero entries and can be multiplied quickly with high
probability, while equation \eqref{eq:wFC} can be solved without parallelization in $O(qnp/s +
qp^2)$ time on average. Throughout this paper we assume $Q$ is renormalized so that $\E
Q^{\top}Q = I_{n}$.  

\section{Tuning Parameter Selection}
\label{sec:tuning-param-select}

Our methods require appropriate selection of $\lambda$ to
achieve good performance.
Presumably, if computations or storage are at a premium, computer-intensive resampling methods such as cross-validation are 
unavailable.  Therefore, we develop methods that rely on a 
corrected training-error estimate of the risk.  These corrections depend crucially on the degrees of freedom.  
Specifically, the {\it degrees of freedom}~\citep{efron1986biased} of a procedure $g:\mathbb{R}^n \rightarrow \mathbb{R}^n$ 
that produces predictions $g(Y) = \hat{Y}$ is
\[
\textrm{df}(g) := \frac{1}{\sigma^2} \sum_{i=1}^n \textrm{Cov}(g_i(Y),\Y_i),
\]
where $\sigma^2 = \mathbb{V}(\Y_i)$.

If the response vector is distributed according to the 
homoskedastic model $Y \sim (\mu, \sigma^2 I_n)$,
then we can decompose the prediction risk of the procedure $g$ as
\[
\textrm{Risk}(g) =\E \norm{g(Y) - \mu}_2^2 = \E \norm{g(Y) - Y}_2^2 -n\sigma^2 + 2 \sigma^2 \textrm{df}(g).
\]
A plug-in estimate of $\textrm{Risk}(g)$
(analogous to $C_p$, \citealp{Mallows1973}) is then
\begin{equation}
\widehat{\textrm{Risk}}(g) = \norm{g(Y) - Y}_2^2 -n\hat\sigma^2 + 2 \hat\sigma^2 \hat{\textrm{df}}(g),
\label{eq:riskEstimate}
\end{equation}
where $\hat{\textrm{df}}(g)$ and  $\hat\sigma^2$ are estimates of $\textrm{df}(g)$ and  $\sigma^2$, respectively.
We discuss strategies for forming $\hat{\textrm{df}}(g)$ in \autoref{sec:dof}.
As for the variance, ordinarily one would use the unbiased estimator
$\hat\sigma^2 = (n-\textrm{df}(g))^{-1}\norm{ (I_n - \Pi_\X )Y}_2^2$, 
where $\Pi_\X  = \X(\X^{\top}\X)^{\dagger}\X^{\top} = UU^{\top}$ is the orthogonal projection onto the column space of $\X$.
However, computing $I_n - \Pi_\X$ is just as expensive as computing the least squares solution $\Pi_\X Y$ itself.
Therefore, to avoid estimating $\sigma^2$, we use generalized cross validation~\citep[][]{GolubHeath1979}
\begin{equation}
\textrm{GCV}(g) = \frac{\norm{g(Y) - Y}_2^2 }{(1 - \textrm{df}(g)/n)^2}.
\end{equation}  
Crucially, GCV does not require a variance estimator, a major advantage.  

\subsection{Estimating the degrees of freedom}
\label{sec:dof}

For any procedure $g$ which is linear in $Y$, that is, there exists
some matrix $\Phi$ which does not depend on $Y$ such that $g(Y) = \Phi
Y$, then 
$\textrm{df}(g) = \tr(\Phi)$, the trace of $\Phi$.
Therefore, computing the exact degrees of freedom for $\hat\beta_{\alpha}(\lambda)$ (that is, the linear combination
estimator with a fixed $\alpha$) is straightforward. In this case,
\[
\X \hat{\beta}_{\alpha}(\lambda) 
= 
\X B(\lambda)\alpha
=
\alpha_1\X \hat{\beta}_{PC}(\lambda) + \alpha_2\X \hat{\beta}_{FC}(\lambda)
=
\Phi_1 Y + \Phi_2 Y=\Phi Y
\]
where $\Phi_1 = \alpha_1\X(\X^{\top}Q^{\top}Q\X + \lambda I)^{-1}\X^{\top}$
and $\Phi_2 = \alpha_2\X(\X^{\top}Q^{\top}Q\X + \lambda
I)^{-1}\X^{\top}Q^\top Q$. So
the degrees of freedom of $\hat\beta_{\alpha}(\lambda)$ is
\begin{equation}
\textrm{df} = 
\alpha_1 \, \tr\left(\X (\X^{\top}Q^{\top}Q\X + \lambda I)^{-1}\X^{\top}Q^{\top}Q\right)
 + 
\alpha_2 \, \tr\left(\X (\X^{\top}Q^{\top}Q\X + \lambda I)^{-1}\X^{\top}\right).
\label{eq:dfForGeneralLinearCombination}
\end{equation}
In particular, both the fully and partially compressed estimators have
simple forms for the degrees of freedom which do not need to be estimated.

For the linear combination estimator when $\alpha$ is estimated, as in
equations \eqref{eq:linearCombinationEstimator} or
\eqref{eq:linearCombinationPredictions}, 
computing the degrees of freedom is more complicated.  This
estimator is nonlinear because both the matrix $B(\lambda)$ and 
the weight vector $\hat\alpha$ are functions of $Y$.
A straighforward estimator of the degrees of freedom for
$\hat\beta_{\hat\alpha}(\lambda)$ is created by plugging $\hat\alpha$
into equation \eqref{eq:dfForGeneralLinearCombination}:
\begin{equation}
\hat{\textrm{df}} =
 \hat{\alpha}_1 \, \tr\left(\X (\X^{\top}Q^{\top}Q\X + \lambda I)^{-1}\X^{\top}Q^{\top}Q\right)
 + 
\hat{\alpha}_2 \, \tr\left(\X (\X^{\top}Q^{\top}Q\X + \lambda I)^{-1}\X^{\top}\right).
\end{equation}
Though this approximation intuitively underestimates the degrees of freedom, the general idea is used for other
nonlinear estimators such as neural networks \citep{Ingrassia2007}.

Alternatively, the degrees of freedom can be computed via Stein's
lemma \citep{stein1981estimation} if we are willing to assume that the
response vector is multivariate normal: $Y \sim N(\mu,\Sigma)$.  Then,
if $g(Y)$ is continuous and almost differentiable in $Y$,
$\textrm{df}(g) = \E [(\nabla \cdot g) (Y)]$,
where $(\nabla \cdot g) (Y) = \sum_{i=1}^n \partial g_i/\partial Y_i$ is the {\it divergence} of $g$.  It immediately follows that
$\widehat{\textrm{df}}(g) = (\nabla \cdot g)(Y)$ is an unbiased
estimator of $\textrm{df}(g)$.  Though the calculus is tedious, the
divergence of the linear combination estimator
can be calculated by repeated applications of the chain rule. 
Unfortunately, we do not yet know how to compute the divergence
without forming $Q^{\top}Q$, which is a large, dense matrix, nor the implications of
the required normality assumption.  Hence,
we consider this a possibly fruitful direction for future research.

\subsection{Computing the path}
\label{sec:computing-path}

In order to select tuning parameters, we need to compute the
estimators quickly for a range of possible $\lambda$. Luckily, this
can be implemented in the same way as with ridge regression. That is we
examine
$
(\X^\top Q^\top Q \X + \lambda I)^{-1} = R(L^2 + \lambda I)^{-1}R^\top,
$
where the singular value decomposition is written $Q\X=SLR^\top$. Therefore, we can take the SVD of $Q\X$ once and
then compute the entire path of solutions for a sequence of
$\lambda$ while only increasing the computational complexity
multiplicatively in the number of $\lambda$ values considered.

\section{Simulations}
\label{sec:simulations}

In this section, we construct 
simulations
to explore when $\hat{\beta}_{\alpha}(\lambda)$ performs
well.

\subsection{Setup}
\label{sec:setup}

To create data, we generate
the design matrix $\X \in \R^{n\times p}$ by independently sampling the
rows from a multivariate normal distribution with mean zero and
covariance matrix $\Sigma$ which has unit variance on the diagonal and
covariance (correlation) $\rho\in\{0.2,\ 0.8\}$ off the diagonal. We then form
$Y = \X\bstar+ \epsilon$,
where $\epsilon_i$ are i.i.d.\ Gaussian with
mean zero and variance $\sigma^2$. In all cases, we take $n=5000$ and
let $p=50,\ 100,\ 250$, or $500$.
%
%

While the performance of our methods depends on all of these
design conditions, we have found that the most important factor is the
structure of $\bstar$. For this reason, we examine three related
structures intended to illustrate the interaction between $\bstar$ and compression. 

In the first case, we take $\b_*\sim\mbox{N}(0,\tau^2 I_p)$. Under
this model, ridge regression is Bayes-optimal for
$\lambda_*=n^{-1}\sigma^2/\tau^2$. We set
$\tau^2=\pi/2$ so that $p^{-1}\E[\norm{\b_*}_1]=1$. Finally, to ensure that
$\lambda_*$ is not too small, we take $\sigma=50$ implying
$\lambda_*\approx 0.32$. The second structure for $\bstar$ is created
to make ridge regression perform poorly: we simply set 
$(\bstar)_j \equiv 1$ for all $1\leq j\leq p$. This scenario is easier for our methods, and
the simulations demonstrate that they outperform both ridge
regression and ordinary least squares (OLS). 
Finally, we choose a middle ground: we take 
$
(\beta_*)_j = (-1)^{j-1}.
$
As in the first case, the coefficients cluster around zero, but here
there is no randomness. In the second and third scenarios, we also use
$\sigma=50$. 

\begin{remark}
The first scenario allows for an easy comparison between our methods
and the optimal algorithm, ridge regression. However, determining
appropriate values for $\tau^2$ and $\sigma^2$ requires some
care. Under this model, 
if $\tau^2$ is large, then $\lambda_*$ will be 
very small so that OLS and ridge regression are
nearly equivalent and our methods will perform relatively poorly. The reason
(see \autoref{sec:theory-results})
is that compression tends to increase baseline variance relative
to ridge regression and therefore requires \emph{more} regularization
and a greater increase in bias. 
On the other 
hand, if $\tau^2$ is small, then $\lambda_*\rightarrow\infty$ and
we will end up shrinking all coefficient estimates to zero all the
time. Compression will exacerbate this effect and will look
artificially good: we should just use the zero estimator. 

To solve this Goldilocks problem and to get nontrivial results in these
simulations, we need to set
$\sigma$ and $\tau$ just right. Note that these
considerations wouldn't be an issue in the high-dimensional,
$p>n$, case since the least squares solution is unavailable. 
Our choice, with
$\lambda_*\approx 0.32$, is intended to be somewhat neutral toward
our methods.
\end{remark}

We examine four different compressed estimators with 
penalization: (1) full compression, (2) partial compression, (3)
a linear combination of the first two, and (4) a convex
combination of the first two. We also use the OLS estimator and the ridge regression estimator. For ridge
regression, we use $\lambda_*$ in the first scenario (and label it
``Bayes'') and choose $\lambda$ by minimizing GCV in the other cases.
We generate 50 training data sets as above with $n=5000$ in all
simulation. 
For the compressed estimators, we examine
three possible values of $q\in\{500,\ 1000,\ 1500\}$. In each case, we
generate $\pre\in\{-1,0,1\}^{q\times n}$ as a ``sparse Bernoulli''
matrix with $s=3$.

\subsection{Estimation error simulations}
\label{sec:estimation-error}

For all four compressed methods, there exist $\lambda$ values which allow the
compressed method to beat ordinary least squares. Some occasionally beat ridge
regression as well depending on the structure of $\beta_*$.
Regularized compression always outperforms unregularized compression
for most $\lambda>0$.
While we have simulated all combinations of $p$ and $\rho$, we only
display results for $p=100$ and $\rho=0.2$ which are
typical. Figures~\ref{fig:estimbox1} to~\ref{fig:estimbox3} show 
boxplots of the results for 
each estimation method across replications. \autoref{fig:estimbox1}
shows the case where $\bstar$ is drawn from a Gaussian
distribution while \autoref{fig:estimbox2} shows $\bstar \in \{-1,1\}^p$
and \autoref{fig:estimbox3} shows $(\bstar)_j \equiv 1,\ \forall j$. Within each
figure, the three panels display different choices of
compression parameter $q$. The $x$-axis shows values
$\lambda$ with the far-right section giving results for OLS and ridge regression. Finally, the $y$-axis is the 
logarithm of the mean squared
estimation error. 
\begin{figure}
\centering
\includegraphics[height=.9\textheight]{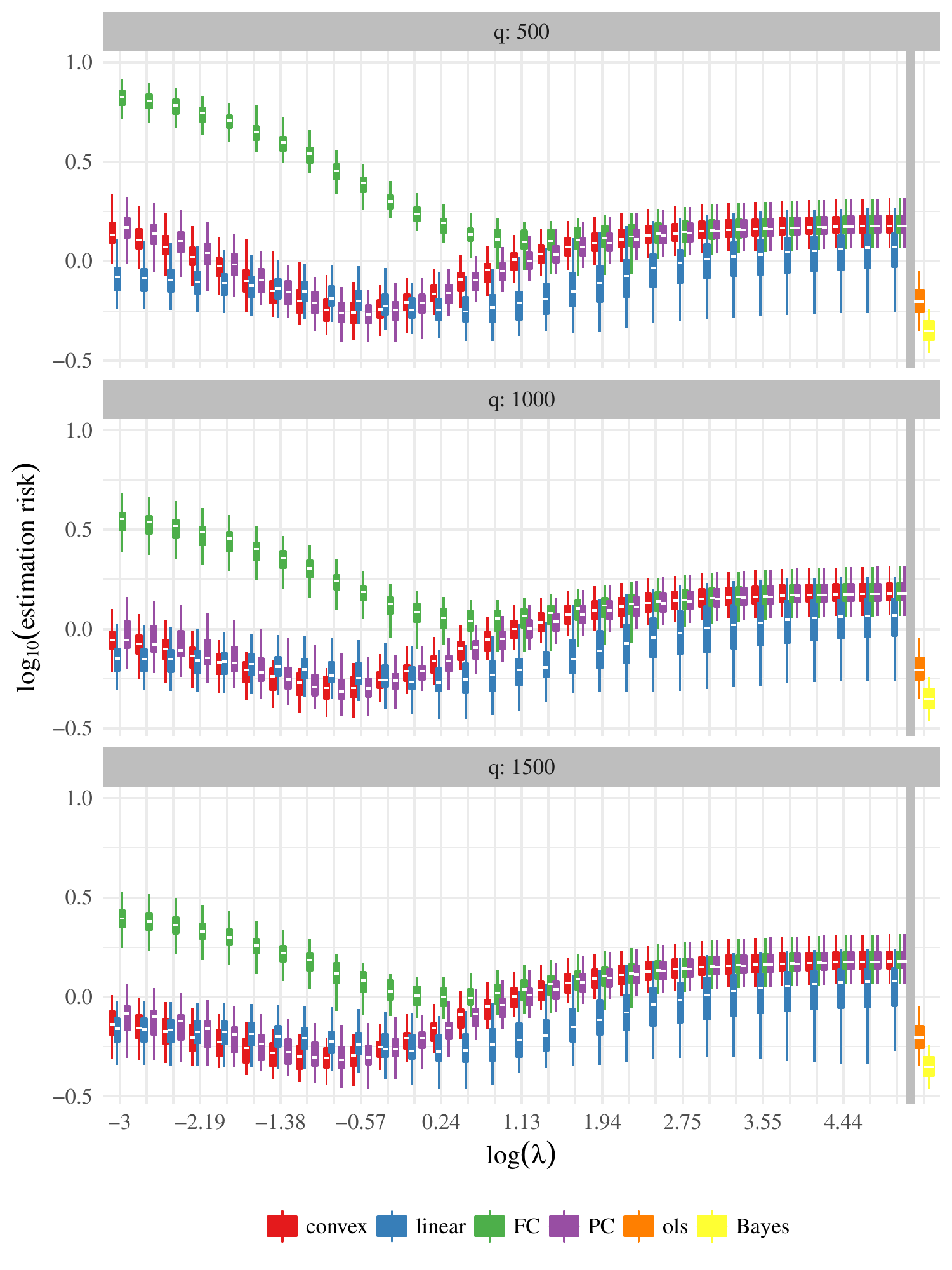}
\caption{Boxplots displaying variance across replications and
  $\lambda$. Here $\bstar$ is Gaussian, $p=100$, and $\rho=0.2$.}
\label{fig:estimbox1}
\end{figure}
\begin{figure}
\centering
\includegraphics[height=.9\textheight]{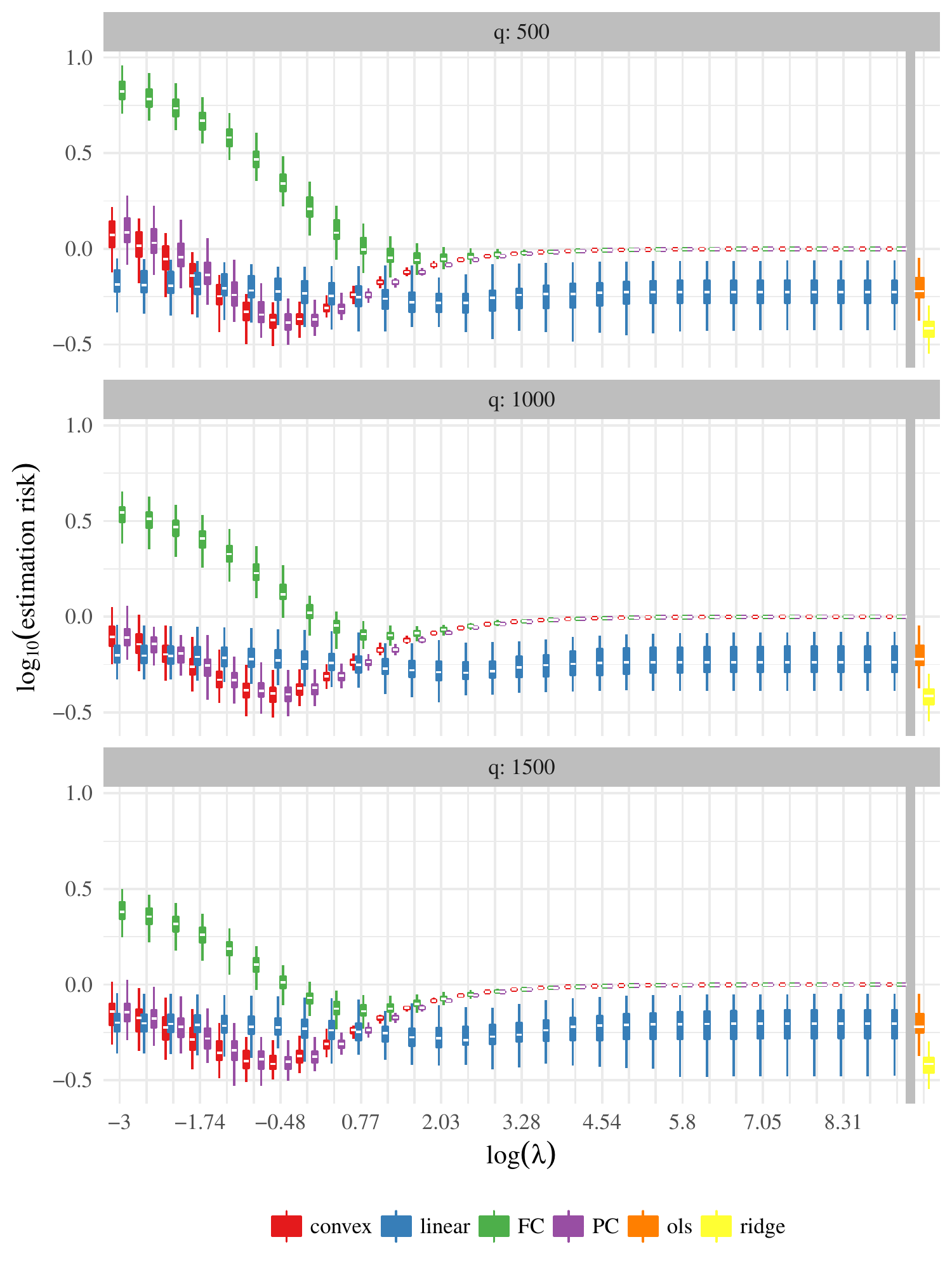}
\caption{Boxplots displaying variance across replications and
  $\lambda$. Here $\bstar\in\{-1,1\}^p$, $p=100$, and $\rho=0.2$.}
\label{fig:estimbox2}
\end{figure}
\begin{figure}
\centering
\includegraphics[height=.9\textheight]{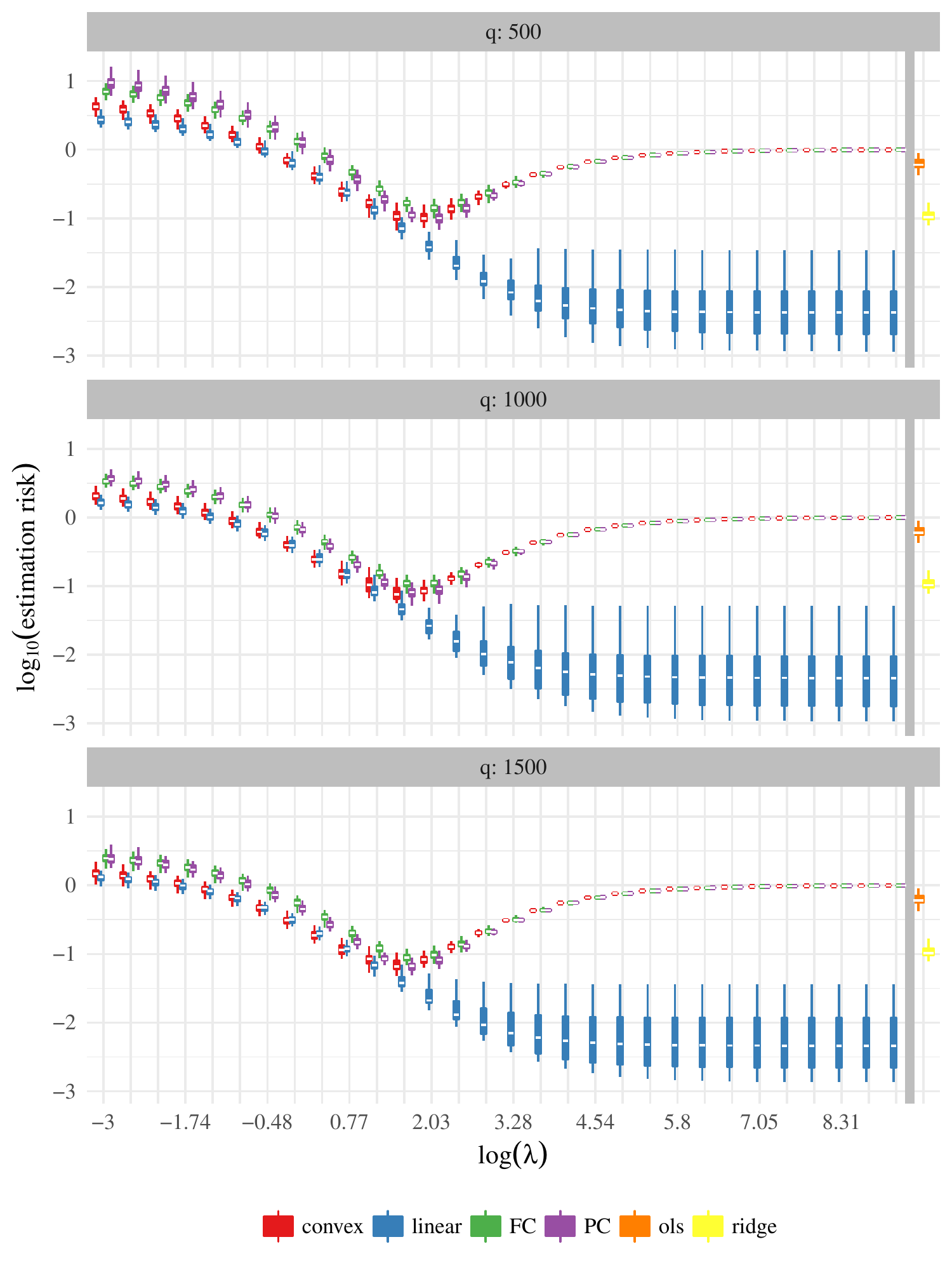}
\caption{Boxplots displaying variance across replications and
  $\lambda$. Here $\bstar\equiv \mathbf{1}$, $p=100$, and $\rho=0.2$.}
\label{fig:estimbox3}
\end{figure}
\begin{figure}
\centering
\includegraphics[height=.9\textheight]{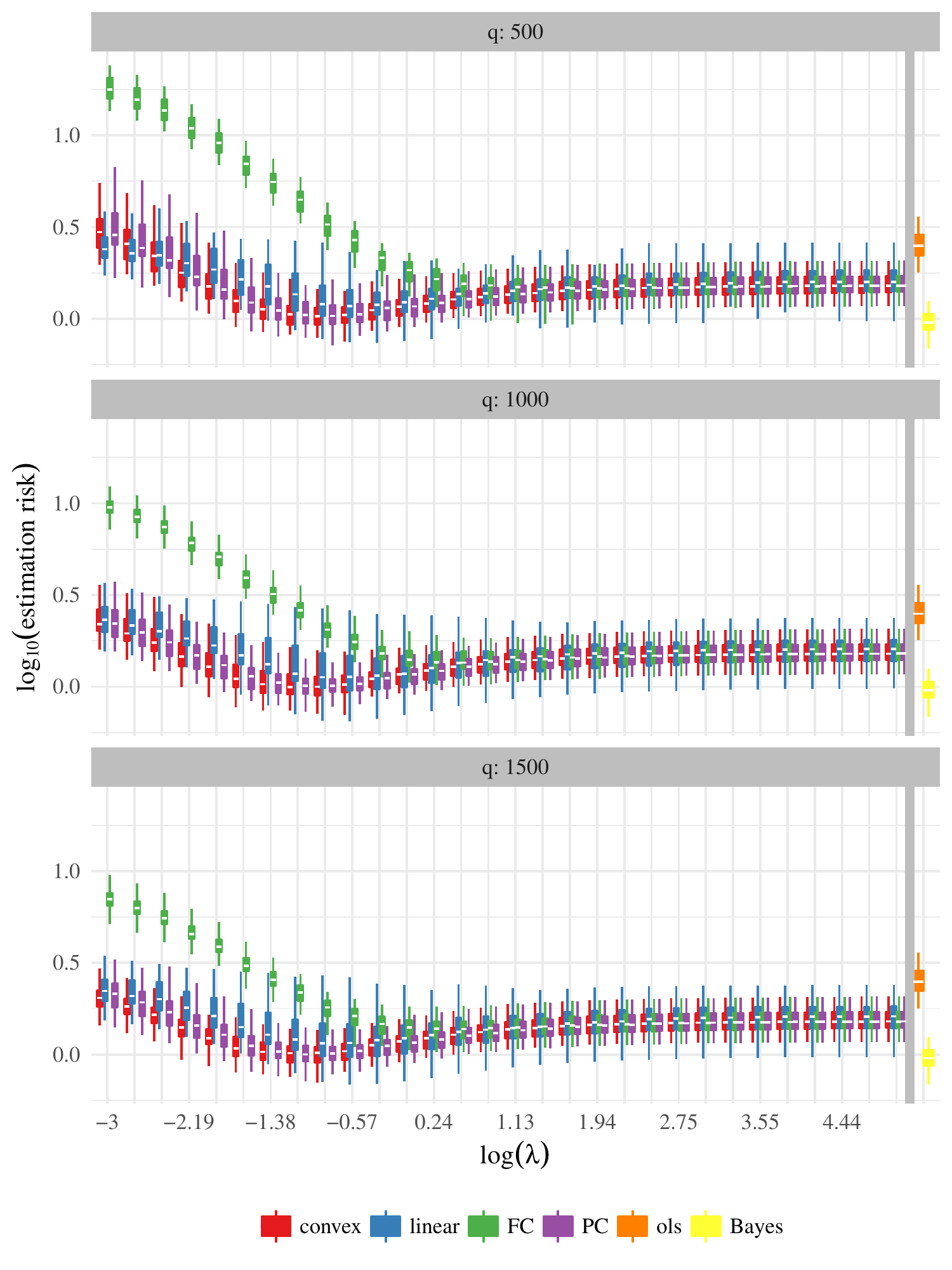}
\caption{Boxplots displaying variance across replications and
  $\lambda$. Here $\bstar$ is Gaussian, $p=100$, and $\rho=0.8$ for
  comparison with \autoref{fig:estimbox1}.}
\label{fig:estimbox4}
\end{figure}

Combining partial and full compression strictly dominates the 
other compressed methods in terms of estimation risk. Furthermore, for
every choice of $q$ and $\beta$, there is a $\lambda$ such that
the linear combination has better estimation risk than OLS. 
When $p$ is small relative to $n$ and the design has low
correlation, only the linear combination and the convex combination
outperform OLS. To compare with the case that the
design has larger
correlation, we also show $p=100$ and $\rho=0.8$ with $(\bstar)_j\equiv 1$
(\autoref{fig:estimbox4}). In this case, all of the compressed methods
outperform OLS for most choices of $\lambda$.

This story holds for
other values of $p$ and $\rho$: (1) regularized compression beats unregularized
compression; (2) when the correlation of the design matrix is
small, there is some level of regularization such that compression
beats OLS; (3) when the correlation of the design
matrix is large, compression beats OLS at
nearly all levels of regularization; (4) the regularized combination
estimators (linear and convex)
are almost always the best; (5) these methods
approach the accuracy of the optimal, full-data Bayes estimator in many cases. 

\begin{remark}
In \autoref{fig:estimbox3}, the performance of the linear
combination estimator is
nearly independent of the choice of $\lambda$, and does
significantly better than the other methods. This behavior is either a
feature or a curse of the method depending on the user's
perspective. Because both full compression and 
partial compression (as well as ridge regression) shrink
coefficients to zero, when $\lambda$ is large, the unconstrained
linear combination will try to use the data to combine two vectors,
both of which are nearly zero, in some optimal way. As long as $\bstar$
is approximately constant, and $\hat{\alpha}$ is unconstrained, this
method can compensate by choosing $\hat{\alpha}$. As
$\lambda\rightarrow\infty$, the solution of
\eqref{eq:linearCombinationEstimator} will adjust and so $\hat\b$ will
actually converge to $(\bstar)_j\equiv 1$ rather than 0 (as the ridge
estimator does). So if the data
analyst believes that $\bstar$ is approximately constant, but nonzero,
this method will work well, even without prior information for the
particular constant. Otherwise, it is safer to use the convex
combination which avoids this pathology. The lack of constraint for
$\hat\alpha$ also has implications
for our approximate degress-of-freedom estimator in
\eqref{eq:dfForGeneralLinearCombination} which we discuss in greater
detail in the next section.
\end{remark}

\subsection{Selecting tuning parameters}
\label{sec:select-tuning-param}

The previous simulation shows that regularized compression can outperform
full-data OLS and compare
favorably with the performance of the Bayes estimator as long as we can choose $\lambda$
well. In this section,
we focus on the case where $\bstar$ has a Gaussian distribution to investigate the empirical
tuning parameter strategies.


We again generate a
training set with $n=5000$, but we also generate an independent test
set with $n=5000$. Then we use the training set to choose
$\lambda_{GCV}$ by generalized cross validation as described in
\autoref{sec:dof}. We also define an optimal (though
unavailable) $\lambda_{test}$ by minimizing the test set prediction error:
\[
\lambda_{test} = \argmin_{\lambda} \frac{1}{n}\norm{Y_{test}-\X_{test} \hat{\beta}(\lambda)}_2^2.
\]

\begin{figure}[t!]
  \centering
  \includegraphics[scale=.8]{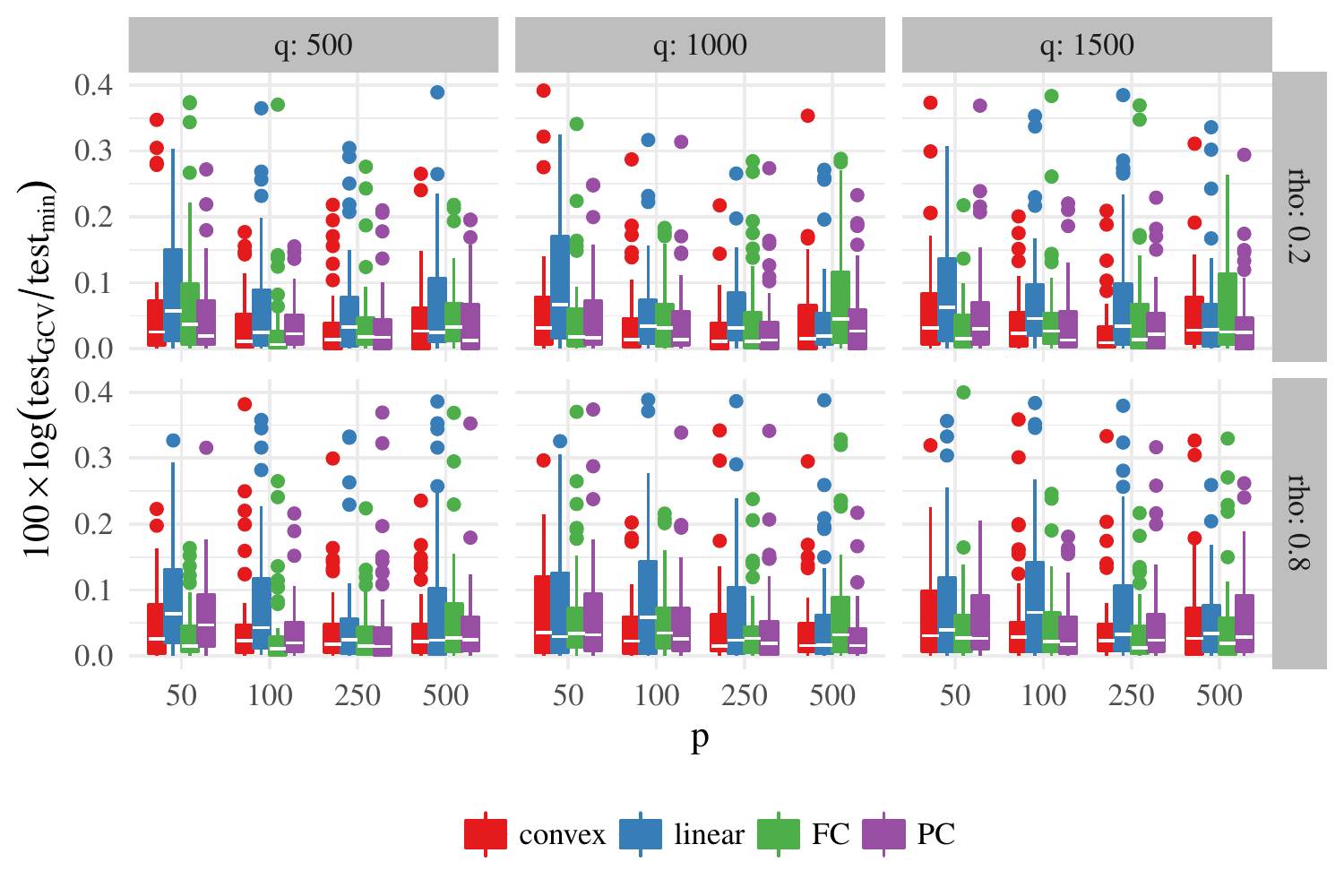}
  \caption{Percentage increase in test error between the tuning parameter chosen by GCV and
    the best tuning parameter we would have chosen with access to the
    same size test set (necessarily greater than 0).}
  \label{fig:tuningparameters2}
\end{figure}


\autoref{fig:tuningparameters2} shows the prediction
risk of the GCV-selected estimates relative to the oracle. That is, we
plot the ratio
\[
\frac{\textrm{test}_{GCV}}{\textrm{test}_{\min}} =
\frac{\norm{Y_{test}-\X_{test} \hat{\beta}(\lambda_{GCV})}_2^2} {\norm{Y_{test}-\X_{test}
    \hat{\beta}(\lambda_{test})}_2^2}  \geq 1.
\]
All methods are
within 1\% of the best test error the majority of the time, but the
convex combination and partial compression tend do the best. Note that
the minimum test error is for the particular method rather than
relative to the best possible test error across all methods. Thus,
while GCV selects the optimal tuning parameter for partial compression
more accurately than for the linear combination, the linear
combination has lower estimation error at its own GCV-selected tuning
parameter.

While \autoref{sec:tuning-param-select} presented two methods for
estimating the degrees of freedom---a simple approximation  as in
equation \eqref{eq:dfForGeneralLinearCombination} or the
divergence---we only present the results for the approximation. As
discussed in \autoref{sec:tuning-param-select}, using the
divergence-based, degrees-of-freedom estimator requires generating an
$n\times n$ matrix, which is computationally infeasible.

\subsection{Overall performance assessments}
\label{sec:over-perf-assessm}

Using the tuning parameter selected by GCV, compressed regression can
perform better than uncompressed regression in some
situations. \autoref{fig:winsestim} examines the performance across
all simulations. Specifically, for each of the 50 training data sets,
we estimate the linear model with each method, choosing $\lambda$ by
GCV if appropriate. We select the method with the lowest estimation
error. We plot the proportion of times each method ``wins''
across all simulation conditions. Generating $\bstar$ from a normal
distribution favors ridge regression, as is to be expected, since
this is the optimal full-data estimator. On the opposite extreme, when
$(\bstar)_j\equiv 1$, the unconstrained linear combination of compressed
estimators is nearly always best. In the middle, each of the
methods works some of the time. 
Overall, the convex combination estimator, that is the estimator given
by \eqref{eq:convexCombinationEstimator}, works well in most cases and has smaller variance
than the unconstrained linear combination. Ordinary least squares
almost never wins.

\begin{figure}[t!]
  \centering
  \includegraphics[scale=.8]{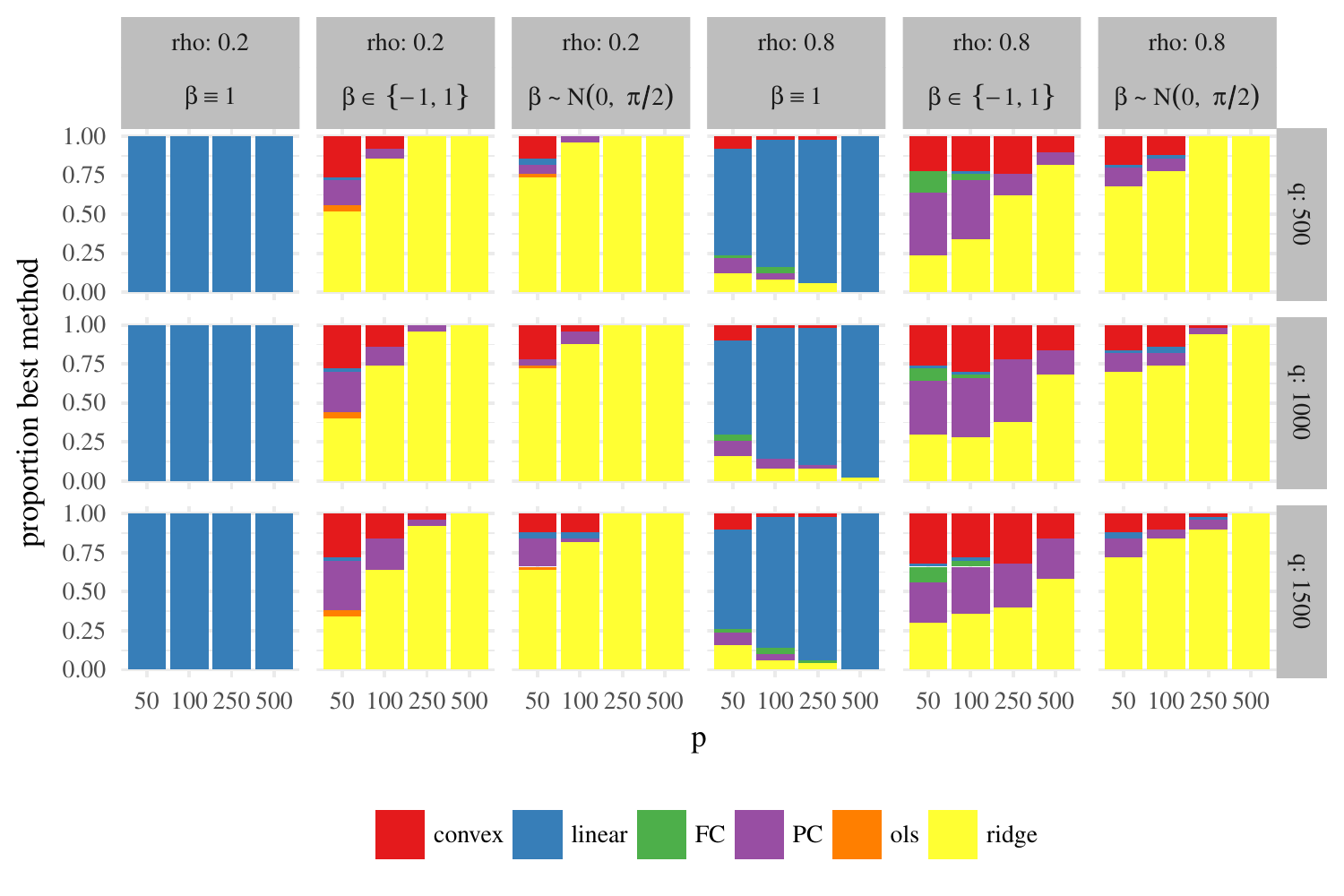}
  \caption{Proportion of simulations a method has the best estimation
    error for each simulation condition.}
  \label{fig:winsestim}
\end{figure}

\subsection{Recommendations}
\label{sec:recommendations}

Even though ridge regression is optimal when $\bstar$ has a Gaussian distribution,
regularized compression can achieve nearly the same prediction error
while using fewer computations and less storage space. Furthermore,
the compressed estimators nearly always outperform ordinary least
squares. \autoref{fig:bestlamtest} displays the prediction risk for
each method when the tuning parameter is chosen by GCV. The difference
between the optimal model and the compressed
approximations is negligible.

\begin{figure}[t!]
  \centering
  \includegraphics[scale=.8]{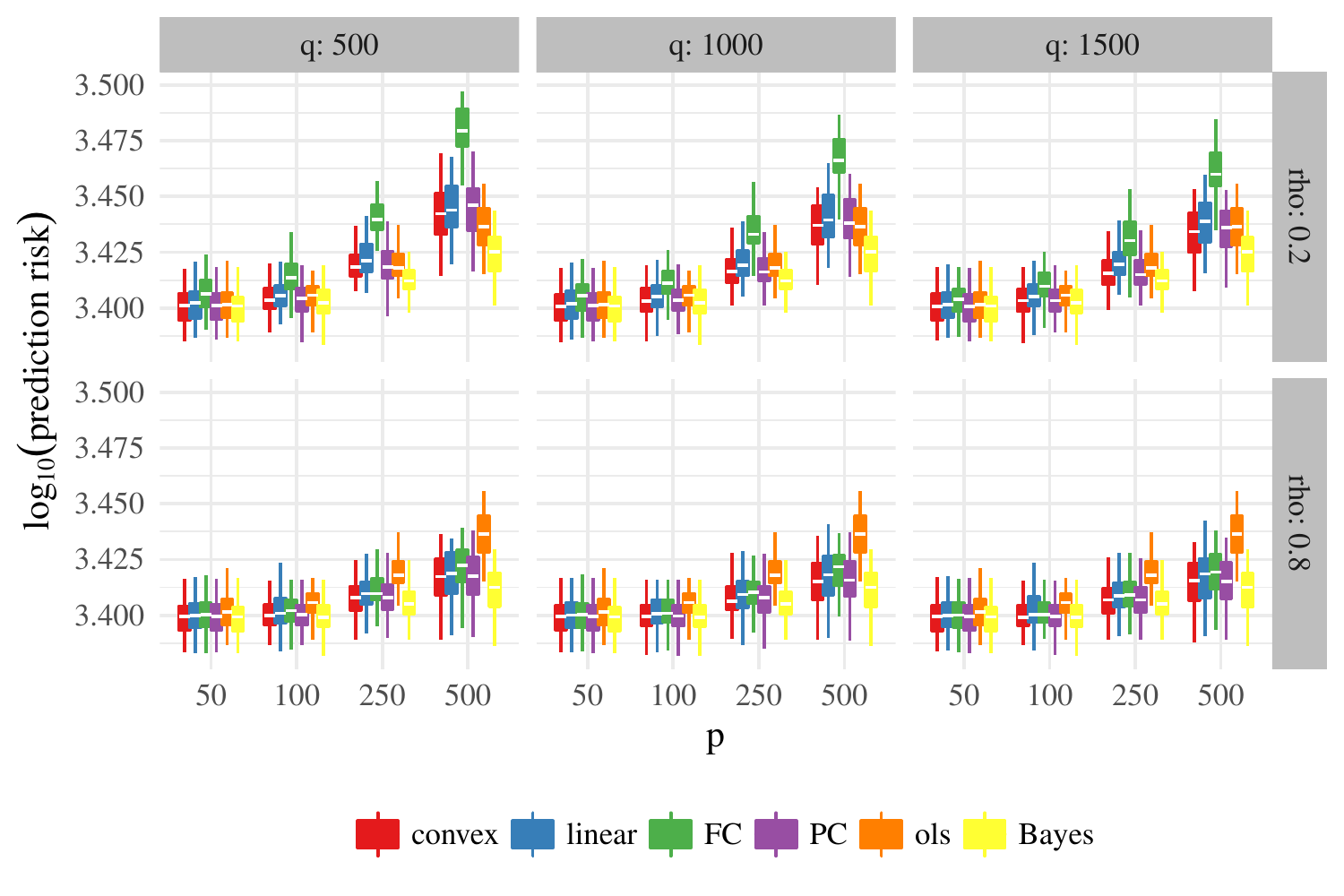}
  \caption{Prediction error on the test set for all methods.}
  \label{fig:bestlamtest}
\end{figure}

\section{Real data examples}
\label{sec:data-example}

We present results for two different types of real data: a collection
of genetics data and an astronomy dataset.

\subsection{Genetics}
\label{sec:genetics}

The first data we examine are a collection of short-read RNA sequences. The
data are publicly available\footnote{From Jun Li:
  \url{http://www3.nd.edu/~jli9/mseq/data_top100.zip}} and were first examined
by \citet{Li2010}, who suggest a Poisson linear model to
predict read counts based on the surrounding nucleotides. An
implementation of this method is provided in the \texttt{R} package
\texttt{mseq}, described by \citet{Li2010} and available from the
\href{https://cran.r-project.org/src/contrib/Archive/mseq/}{CRAN
  archive}.

In all, there are eight data files from three research groups. Three
datasets are due to \citet{MortazaviWilliams2008}, which
mapped mouse transcriptomes from brain, liver, and skeletal muscle
tissues. \citet{WangSandberg2008} collected data from
15 different human tissues which have been merged into three groups
based on tissue similarities. Finally, 
\citet{CloonanForrest2008} examined RNA sequences from mouse
embryonic stem cells and embryoid bodies. In all cases, we use the top
100 highly-expressed genes as well as the surrounding sequences to
predict expression counts as in \citep{Li2010}. \autoref{tab:genesN} shows the sample size $n$ for each of
the eight data sets.

Following \citet{Li2010} and \citet{DalpiazHe2013}, we examine each of these
datasets separately. In order to build the model, we must select how
many surrounding nucleotides to use for prediction. As
nucleotides are factors (taking levels C,T,A,G), a window of
$k$ surrounding nucleotides will give $p=3(k+1)+1$ predictors of
which one is the intercept. We could also use dinucleotide
pairs (or higher interactions), as in \citet{Li2010}, resulting in
$p=15(k+1)+1$ predictors. For our illustration, we follow
\citet{MaMahoney2015}, who also apply different compressed
linear regression methods to these data, and use $k=39$. 
\begin{table}[t!]
\centering
\begin{tabular}{ccccccccc}
  \hline\hline
  dataset & B1 & B2 & B3 & G1 & G2 & W1 & W2 & W3 \\  
  $n$ & 157614 &125056 & 103394&51751&64966&146828&171776 &143570\\ 
   \hline
\end{tabular}
\caption{Number of observations for each of the 8 genetics data sets.}
\label{tab:genesN}
\end{table}

The results of our analysis are shown in
\autoref{fig:genes}. For each dataset, which we denote by the
first letter of the senior author's last name (W, B, and G respectively) followed
by a number, we split the data randomly into
75\% training data and 25\% testing data. We then compress the training set using
$q=10000$ and $q=20000$. We
apply each of the regularized compressed methods, choosing $\lambda$ by
generalized cross validation and then evaluate the estimators by
making predictions on the test set. We repeat this procedure 10 times and present
the average of the log test error relative to OLS. Across data sets, $q=10000$ results in data
reductions between 74\% and 93\% (meaning $0.26 \geq q/n \geq 0.07$) while $q=20000$ gives reductions
between 48\% and 84\%. 
\begin{figure}[t!]
  \centering
  \includegraphics[scale=.8]{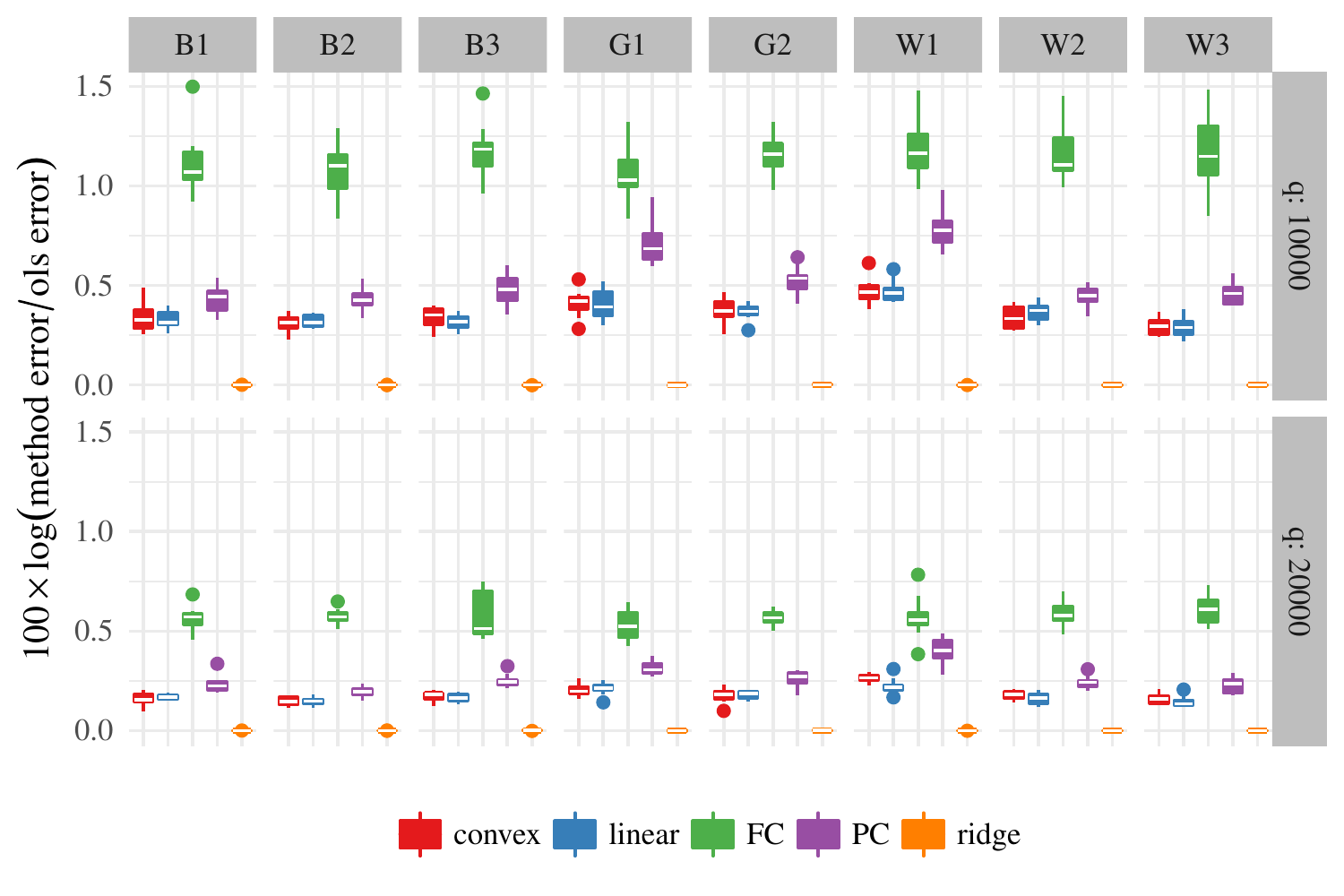}
  \caption{Results of each method on 10 replications of training test
    splits on each of the eight genetics data sets. The results are
    percentage increase in test error relative to ordinary least
    squares. }
  \label{fig:genes}
\end{figure}

For these data, ridge and OLS give equivalent test set
performance (differing by less than .001\%) across all
data sets. The similarity between these two estimators suggests that
compression will be hard-pressed to yield 
improvements since the signal-to-noise ratio is so high. Additionally,
previous analyses have found that the coefficients are roughly centered around
zero with most quite small. While none of
our methods are able to beat OLS, their performance is not
much worse. The worst method is always full compression. The linear
combination and the convex combination are nearly equivalent, while partial
compression is just slightly worse. Even for full compression, its
worst performance across all data sets and over both values of $q$ is
less than 1.5\% worse than OLS. So even for the worst performing method, large amounts of
compression result in a negligible increase in test set error.

\subsection{Galaxies}
\label{sec:galaxies}



To create another data set with larger coefficients possibly
not centered around zero, we examine a collection of galaxies and
attempt to predict their redshifts, a measurement of how far the
galaxy is from earth, from their flux measurements, the intensity of
light at different wavelengths. We
downloaded $n=5000$ random galaxies from the
Sloan Digital Sky Survey
12~\citep{AlamAlbareti2015}. Following~\citet{RichardsFreeman2009}, we
restricted our sample to contain only galaxies with estimated 
redshift $z>0.05$ and plate number less than 7000 (currently, larger
plate numbers may contain galaxies without observed flux
measurements). We also removed galaxies which have been flagged as
having unreliable redshift estimates. While 5000 galaxies is quite a
bit smaller than the genetics data we used above, we note that this is just
a tiny subsample of the nearly 1.5 million currently available. Each galaxy has 3693 flux measurements measured at
wavelengths between 3650\AA{} and 10400\AA, but following previous
analyses, we truncate those below 4376\AA{} and above 10247\AA{}. Finally, because the
measurements are quite noisy, we smooth the flux measurements using a
regression spline with 125 equally spaced knots on the log scale. 
We use the 125 spline
coefficients from the smoothed
versions for all 5000 galaxies as our design matrix.

As above, we randomly divide the data into 75\% training
data and 25\% testing data, estimate each method on the training set,
and predict the test data. We repeat the experiment 10 times using
$q=1000$ and $q=2000$. The results are shown in
\autoref{fig:galaxyres}. 
\begin{figure}[t!]
  \centering
  \includegraphics[scale=.8]{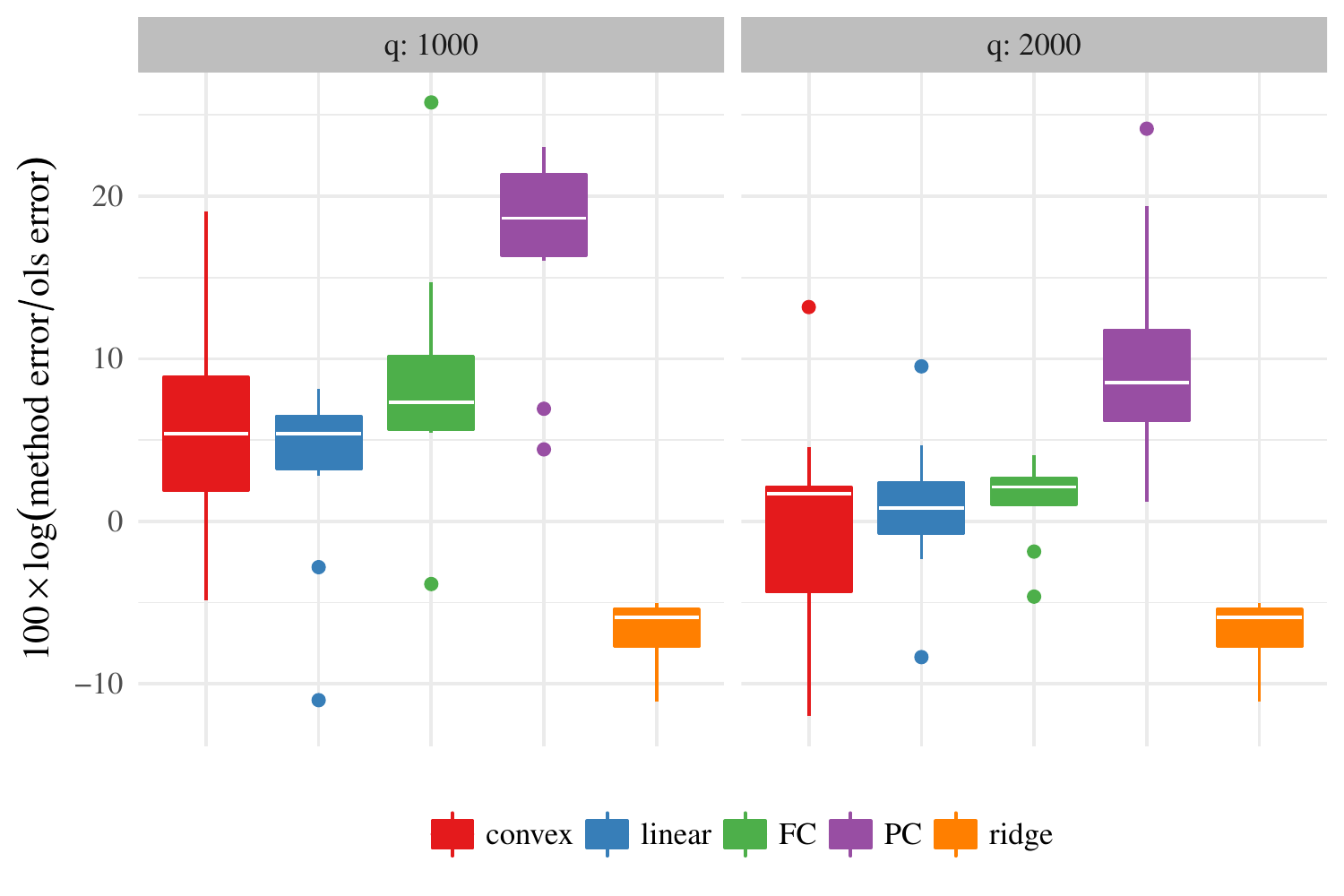}
  \caption{Results of each method on 10 replications of training test
    splits on each of the eight genetics data sets. The results are
    percentage increase in test error relative to ordinary least
    squares.}
  \label{fig:galaxyres}
\end{figure}
In this case, ridge regression has better performance than ordinary
least squares on every replication, improving test error between 5 and
10\%. None of the compressed methods do quite this well. When
$q=1000$, the convex combination and linear combination are generally
between 3 and 8\% worse than OLS. Interestingly, in this case, full compression
is not much different, but partial compression is significantly worse
with test error between about 16\% and 21\% worse than OLS. With
$q=2000$, the convex combination is again the best, with 
about half of the replications outperforming OLS,
while the linear combination is not far behind. The convex combination
had better test error than ridge regression in one of the ten replications.

\section{Theoretical analysis}
\label{sec:theory-results}

To develop a better understanding of the relationship between
the compressed regression methods proposed here and standard full-data
techniques, we derive expressions for the expectation and variance of
full and partial compression estimators as well as their bias and
variance. For comparison, we first present the standard analogues for
ridge regression.

\subsection{Standard results for ridge regression}
\label{sec:stand-results-ridge}

Write the singular value decomposition of
$\X=UDV^\top$. Then define the ridge regression estimator of $\bstar$
as in equation \eqref{eq:ridge}.

The bias and variance of the ridge regression estimator conditional on
the design matrix are given in the first result.
\begin{lemma}
  \begin{align*}
    \E[\bstar - \hat\b_{ridge}(\lambda)\given \X] 
    &= [I - (\X^\top\X + \lambda I_p)^{-1}\X^\top\X]\bstar
    = (I + \lambda^{-1} \X^\top\X)^{-1}\bstar\\
    &= \lambda V(D^2 + \lambda I_p)^{-1}V^\top\bstar.\\
    \V[\hat\b_{ridge}(\lambda) \given \X] 
    &= \sigma^2 (\X^\top\X + \lambda I_p)^{-1} \X^\top \X (\X^\top\X +
      \lambda I_p)^{-1}\\
    &= \sigma^2V(D^2 + \lambda I_p)^{-1} D^2 (D^2 + \lambda I_p)^{-1} V^\top.
  \end{align*}
\end{lemma}

A standard corollary gives the squared bias and trace of the variance
of the ridge regression estimator conditional on the design.
\begin{corollary}
  \begin{align*}
    \bias^2\left(\hat\b_{ridge}(\lambda)\given\X\right)
    &= \lambda^2 \bstar^\top V(D^2 + \lambda I_p)^{-2} V^\top \bstar.\\
    \tr \left( \V[\hat\b_{ridge}(\lambda) \given \X]\right)
    &= \sigma^2 \sum_{j=1}^p \frac{d_j^2} {(d_j^2 + \lambda)^2}.
  \end{align*}
\end{corollary}

In the next sections, we will derive approximations to these
quantities for the fully and partially compressed
ridge regression estimators of $\bstar$.

\subsection{Mean and variance of the compressed estimators}
\label{sec:results-via-taylor}

Because all of our estimators depend on $(\X^\top Q^\top Q \X +
\lambda I_p)^{-1}$, a generally intractable quantity, we derive
approximate results via a first order
Taylor expansion of the estimator with respect to the matrix $Q^\top
Q$. The proofs as well as intermediary results are included
in the Supplementary Material.

Following \citet{MaMahoney2015}, we use the Taylor expansion of
$\hat\b$ as a function of $A:=\frac{s}{q}Q^\top Q$ around $I_n$ to derive results conditional
on $Y$ and $\X$ (taking expectations over $Q$) as well as results unconditional
on $Y$. The first case reflects the randomness in the compression
algorithm relative to the more computationally demanding ridge
regression. The second is useful for comparing the compressed
procedures with ridge regression by including randomness introduced
through the data generating process and through the compression
algorithm.  In all cases, these results are 
conditional on the design 
matrix as was the case 
above. For convenience of expression, define
\begin{align*}
  M &:= (\X^\top  \X + \lambda 
      I_p)^{-1} \X^\top = V(D^2+\lambda I_p)^{-1} D U^\top, \textrm{ and}\\
  H &:=\X(\X^\top  \X + \lambda 
      I_p)^{-1} \X^\top = UD(D^2+\lambda I_p)^{-1} D U^\top = \X M.
\end{align*}
Finally, for these results we will assume that $q=c n$ for some
$0<c\leq 1$ which is fixed. We discuss this assumption further in the
remark below.

\begin{theorem}
  \label{thm:coefFC}
  For full compression,
  \begin{align}
    \E[\hat\b_{FC}(\lambda)\given \X, Y] 
    &= \hat\b_{ridge}(\lambda) + o_P(1)\\
    \V[\hat\b_{FC}(\lambda)\given \X, Y]
    &=  \frac{(s-2)_+}{q} M\hat{e}\hat{e}^\top M^\top + \frac{1}{q} 
                                \hat{e}^\top\hat{e} MM^\top  +
     o_P(1),
  \end{align}
  where $\hat e=(I-H)Y$. Furthermore,
  \begin{align}
    \E[\hat\b_{FC}(\lambda)\given \X] 
    &= [I-\lambda(\X^\top\X +\lambda I)^{-1}]\b + o_P(1)\\
    \V[\hat\b_{FC}(\lambda)\given \X] 
    &= \sigma^2V(D^2+\lambda I)^{-1}D^2(D^2+\lambda I)^{-1}V^\top+\\ 
    &\quad + \frac{(s-2)_+}{q} M
      (I_n-H)(\X\bstar\bstar^\top\X^\top+\sigma^2I_n)(I_n-H)M^\top\\
    &\quad+ \frac{1}{q} \left( \sigma^2\tr((I_n-H)^2)MM^\top
+ \bstar^\top
      \X^\top(I_n-H)^2\X\bstar MM^\top\right) + o_P(1),
  \end{align}
  where $\hat Y = HY$.

\end{theorem}
\begin{theorem}
  \label{thm:coefPC}
  For partial compression,
  \begin{align}
    \E[\hat\b_{PC}(\lambda)\given \X, Y] 
    &= \hat\b_{ridge}(\lambda) + o_P(1)\\
    \V[\hat\b_{PC}(\lambda)\given \X, Y] &=  \frac{(s-2)_+}{q}M\hat{Y}
                                \hat{Y}^\top M^\top + \frac{1}{q} 
                                \hat{Y}^\top\hat{Y} MM^\top
                                             +o_P(1),
  \end{align}
  where $\hat{Y} = H Y$.
  Additionally,
  \begin{align}
    \E[\hat\b_{PC}(\lambda)\given \X] 
    &= [I-\lambda(\X^\top\X +\lambda I)^{-1}]\b+ o_P(1)\\
    \V[\hat\b_{PC}(\lambda)\given \X] 
    &= \sigma^2V(D^2+\lambda I)^{-1}D^2(D^2+\lambda
                      I)^{-1}V^\top\\
    &\quad +\frac{(s-2)_+}{q}M
      H(\X\bstar\bstar^\top\X^\top+\sigma^2I_n)H M^\top \\
    &\quad+ \frac{\sigma^2\tr(H^2)}{q} MM^\top  +\frac{1}{q}
      \bstar^\top \X^\top H^2 \X \bstar MM^\top  + o_P(1).
  \end{align}
\end{theorem}

\begin{remark}
In each
expression above, the Taylor series is valid whenever higher-order
terms are small. In our case, the higher-order terms are
$o_P(\norm{A-I}^2)$ under the expansion, for some matrix norm
$\norm{\cdot}$. Here $o_P(\cdot)$ is with 
respect to the randomness in $A$ through
$Q$. As $Q$ is independent of the data, one can examine how
large these deviations are likely to be. In particular, using results
similar to the Tracy-Widom law \citep[][Proposition
2.4]{RudelsonVershynin2010}, one can show that
$\norm{A-I}=O_P(\sqrt{n/q})$. Therefore, taking $q=c n$ for some $0<
c\leq 1$ means that the remainder is $o_P(1)$. This is in contrast
with results of \citet{MaMahoney2015}, for two reasons: (1) their sampling mechanism
is allowed to depend on the data where ours is not and (2) they can
lose rank from the compression.
In our case, $(\X^\top Q^\top Q\X + \lambda I)$ is full rank for all $\lambda>0$ regardless of the
rank of  $\X^\top Q^\top Q\X$, so the Taylor series is always valid.
\end{remark}

On average, both procedures are the same as ridge regression, up to
higher order remainder terms. The variance, however, is larger for the
same value of $\lambda$. Finally, we note that these results make explicit another tradeoff
between computation and estimation: taking $s\leq 2$ eliminates one
term in the variance expansion completely at the expense of denser $Q$.

In order to compare our procedure to ridge regression more directly, we
examine the Taylor expansions of the
mean squared error (rather than expanding the estimator and then
taking expectations) in the next section.

\subsection{Mean squared error performance}

Rather than looking at the estimated coefficients, we could also look
at the mean squared error (MSE) directly. If we ignore remainder
terms, we can try to minimize the
MSE over $\lambda$ and evaluate the quality of the resulting oracle
estimator. As suggested in our simulations, there exists a $\lambda_*$
such that compressed ridge regression actually achieves lower MSE than
ridge regression does. 

\begin{theorem}
  \label{thm:bias2taylor}
  The squared-bias of the fully compressed estimator is:
  \begin{align*}
    \bias^2\left(\hat\b_{FC} \given\X, Q\right)
    &= \lambda^2 \bstar^\top V(D^2 + \lambda I)^{-2} V^\top \bstar+
      2\bstar^\top(M\X-I) M (A-I) (I -  
      H)\X\bstar +o_P(1)\\
    \bias^2\left(\hat\b_{FC}\given \X\right)
    &= \lambda^2 \bstar^\top V(D^2 +
      \lambda I)^{-2} V^\top \bstar +
     o_P(1).
  \end{align*}
  The squared-bias of the partially compressed estimator is:
  \begin{align*}
    \bias^2\left(\hat\b_{PC} \given \X, Q\right)
    &=\lambda^2 \bstar^\top V(D^2 + \lambda I)^{-2} V^\top \bstar + 
      2\bstar^\top(I-M\X) M
      (A-I) H\X\bstar + o_P(1)\\
    \bias^2\left(\hat\b_{PC}\given \X\right) 
    &= \lambda^2 \bstar^\top V(D^2 + \lambda I)^{-2} V^\top \bstar +
      o_P(1).
  \end{align*}
\end{theorem}
Note that, ignoring
the remainder, the squared bias of both fully compressed and partially
compressed estimators when averaged over the compression is the same as that of ridge regression.

\begin{theorem}
  \label{thm:trace-variance-taylor}
  The variance of the fully compressed estimator is:
  \begin{align*}
    \tr\left(\V\left[\hat\b_{FC} \given \X, Q\right]\right) 
    &= \sigma^2 \sum_{j=1}^p \frac{d_j^2} {(d_j^2 + \lambda)^2}+ 2
      \tr\left(M(I_n-H)(A-I_n)M^\top\right)+ o_P(1)\\
    \tr\left(\V\left[\hat\b_{FC} \given \X\right]\right) 
    &= \sigma^2 \sum_{j=1}^p \frac{d_j^2} {(d_j^2 + \lambda)^2} +
      o_P(1)\\
   &\quad+     \frac{\lambda^2(s-2)_+}{q}\bstar^\top V D^2(D^2+\lambda
                             I)^{-4}D^2V^\top \bstar\\
  &\quad + \frac{\lambda^2}{q}\bstar^\top VD(D^2+\lambda I)^{-2}DV^\top\bstar\sum_{j=1}^p \frac{d_j^2} {(d_j^2 + \lambda)^2}.
  \end{align*}
  The variance of the partially compressed estimator is:
  \begin{align*}
  \tr\left(\V\left[\hat\b_{PC} \given \X, Q\right]\right) 
    &= \sigma^2 \sum_{j=1}^p \frac{d_j^2} {(d_j^2 + \lambda)^2} - 2
      \tr\left(MH(I_n-A)M^\top\right)+ o_P(1)\\
  \tr\left(\V\left[\hat\b_{PC} \given \X\right]\right) 
    &= \sigma^2 \sum_{j=1}^p \frac{d_j^2} {(d_j^2 + \lambda)^2} +
      o_P(1)\\
  &\quad +         \frac{(s-2)_+}{q}\bstar^\top VD^2(D^2+\lambda I)^{-2}D^2 V^\top\bstar\\
  & \quad+ \frac{1}{q}\bstar^\top VD^3(D^2+\lambda I)^{-2}D^3 V^\top\bstar\sum_{j=1}^p \frac{d_j^2} {(d_j^2 + \lambda)^2}.
  \end{align*}

\end{theorem}

\begin{corollary}
\label{cor:lambda-star}
  Suppose $\X$ is such that $\X^\top \X = nI_p$,
  $b^2 := \norm{\bstar}_2^2$, and $\theta:= \lambda/n$. Then
    \begin{align*}
    \textrm{MSE}(\hat\beta_{ridge})
    &= b^2\left(\frac{\theta}{1+\theta}\right)^2 +
      \frac{p\sigma^2}{n(1+\theta)^2}\\ 
    \textrm{MSE}(\hat\beta_{FC})
    &= b^2\left(\frac{\theta}{1+\theta}\right)^2 +
      \frac{p\sigma^2}{n(1+\theta)^2}+
      \frac{b^2p\theta^2(s-2)_+}{q(1+\theta)^4} +
      \frac{p^2\theta^2b^2}{q(1+\theta)^4}\\ 
    \textrm{MSE}(\hat\beta_{PC})
    &= b^2\left(\frac{\theta}{1+\theta}\right)^2 +
      \frac{p\sigma^2}{n(1+\theta)^2}+
      \frac{p(s-2)_+b^2}{q(1+\theta)^2} + \frac{pb^2}{q(1+\theta)^4}.
  \end{align*}
\end{corollary}

Hence, for ridge, the optimal $\theta_* = \sigma^2 p/(n b^2)$ and
$\lambda_* = \sigma^2 p/b^2$.  For the other methods, the MSE can be
minimized numerically, but we have, so far, been unable to find an
analytic expression.

\section{Conclusion}
\label{sec:conclusion}

%

In this paper, we propose and explore a broad family of compressed, regularized
linear model estimators created by generalizing
a commonly used approximation procedure which we notate $\hat\beta_{FC}$ (defined in equation \ref{eq:fullCompression}).
We show that $\hat\beta_{FC}$  must indeed perform worse than the least squares solution.  We suggest 
using
$\hat{\b}_\alpha(\lambda)$, defined in equation \eqref{eq:mostGeneral} instead.  As $\hat{\b}_\alpha(\lambda)$
has a tuning parameter $\lambda$, we discuss methods for choosing it in a data 
dependent and computationally feasible way.

We find that, in particular cases, $\hat{\b}_\alpha(\lambda)$ can perform
better than full-data OLS and full-data ridge regression, and that this
statistical performance gain is accompanied by a computational one.  Admittedly, $\hat{\b}_\alpha(\lambda)$ does
not perform as well in our real data examples, however, it appears that these data sets have relatively low 
variance, and hence, including bias does not improve performance as much as if there were larger variance.

Interesting future work would examine the divergence based tuning parameter selection method.
This approach, while promising, seemingly requires the computation of the dense, large matrix $Q^{\top}Q$, and 
hence, is computationally infeasible.  Additionally, other forms for the compression matrix $Q$ should be investigated
to examine their statistical impact. Finally, while ridge penalties
are amenable to theoretical analysis because of their closed form
solution, it is worthwhile to examine other penalty functions and
other losses, such as generalized linear models.







\bibliographystyle{agsm}
\bibliography{../plsRefs}


\clearpage

\appendix

\section{Supplementary material}
\label{sec:suppl-mater-app}

\subsection{Background results}
\label{sec:background-results}

\begin{lemma}
  \[
  \diag(\vecfun{I_n}) = \sum_{j=0}^{n-1} e_{1+j(n+1)}e_{1+j(n+1)}^\top
  \]
  where $e_j$ is an $n^2$ vector with a 1 in the $j^{th}$ position and
  zeros otherwise.
\end{lemma}

\begin{lemma}
  \[
  (B^\top \otimes A) \diag(\vecfun{I_n}) (B \otimes A^\top)  = \sum_{j=1}^{n}\vecfun{A E_{jj}B} \vecfun{AE_{jj}B}^\top
  \]
where $E_{jj}$ is the $n\times n$ matrix with a 1 in the $jj$ entry
and zeros otherwise.
\end{lemma}

\begin{proof}
  \begin{align}
    &(B^\top \otimes A) \diag(\vecfun{I_n}) (B \otimes A^\top) 
    = \sum_{j=0}^{n-1}(B^\top \otimes A) e_{1+j(n+1)}e_{1+j(n+1)}^\top
      (B \otimes A^\top)\\
    &= \sum_{j=1}^{n}\vecfun{A E_{jj}B} \vecfun{AE_{jj}B}^\top.
  \end{align}
\end{proof}

\begin{corollary}
  If $b$ is a vector, then 
  \[
  (b^\top \otimes A) \diag(\vecfun{I_n}) (b \otimes A^\top)  =
  Abb^\top A^\top.
  \]
\end{corollary}

\begin{lemma}
  \label{lem:ridge-residual-expect}
  Let $\hat{e} = (I_n-H)Y$ for $H=\X(\X^\top\X + \lambda
  I_p)\X^\top=UD(D^2+\lambda I_p)^{-1}DU^\top$. Then,
  \begin{align*}
    \E[\hat{e}\hat{e}^\top] 
    &= (I_n-H)(\X\bstar\bstar^\top\X^\top+\sigma^2I_n)(I_n-H)\\
    \E[\hat{e}^\top\hat{e}]
    &= \sigma^2\tr((I_n-H)^2)+\bstar^\top \X^\top(I_n-H)^2\X\bstar.
  \end{align*}

\end{lemma}

\begin{lemma}
  \label{lem:ridge-hat-expect}
  Let $\hat{Y} = HY$. Then,
  \begin{align*}
    \E[\hat{Y}\hat{Y}^\top] 
    &= H(\X\bstar\bstar^\top\X^\top+\sigma^2I_n)H\\
    \E[\hat{Y}^\top\hat{Y}]
    &= \sigma^2\tr(H^2)+\bstar^\top \X^\top H^2\X\bstar.
  \end{align*}

\end{lemma}

\section{Proofs}

\begin{lemma}
  \label{lem:q_ij-distr-indep}
  For $q_{ij}$ distributed independently as sparse Bernoulli random
  variables, we have
  \begin{align}
    \E\left[\frac{s}{q} Q^\top Q\right] 
    &= I_{n}\\
    \V\left[ \vecfun{\frac{s}{q} Q^\top Q}\right] &=\frac{(s-3)_+}{q} \diag(\vecfun{I_n}) +
  \frac{1}{q} I_{n^2}+\frac{1}{q} K_{nn}
  \end{align}
  where $K_{nn}$ is
  the $(nn,nn)$ commutation matrix: $K_{nn}$ is the unique
  matrix such that, for any
  $n\times n$ matrix $A$, $\vecfun{A^\top} = K_{nn} \vecfun{A}$
  \citep[see][]{MagnusNeudecker1979}. 
\end{lemma}

\begin{proof}[Proof of \autoref{lem:q_ij-distr-indep}]
  We have
  \begin{align}
    \E[q_i^\top q_j] &= \sum_{k=1}^q \E[q_{ik}q_{jk}] = \sum_{k=1}^q
                       \frac{1}{s}I(i=j) = \frac{q}{s}.
                       \intertext{Therefore,}
                       \E[Q^\top Q] &= \frac{q}{s} I_{n}.
  \end{align}
  Now, for the variance, note that each element of $\vecfun{Q^\top
    Q}$ is equal to $q_i^\top q_j$ for appropriate $i,\ j \in
  \{1,\ldots,n\}$. Therefore, the $(ij,kl)$-element of the $(n^2\times
  n^2)$-variance
  matrix is given by $\E[q_i^\top q_j q_\ell^\top q_m]
  -(\E[Q^\top Q])_{i,j}(\E[Q^\top Q])_{\ell,m}.$ 
  \begin{align}
    \E[q_i^\top q_j q_\ell^\top q_m] 
    &= \E\left[\sum_{k=1}^qq_{ik}q_{jk}\sum_{s=1}^q q_{\ell s} q_{m
      s}\right]\\ 
    &= \E\left[(q_{i1}q_{j1}+\cdots+q_{iq}q_{jq})(q_{\ell
      1}q_{m1}+\cdots+q_{\ell q}q_{mq})\right]\\
    &= \begin{cases}
      \frac{q}{s}+\frac{q^2-q}{s^2} & i=j=\ell=m\\
      \frac{q^2}{s^2} & i=j,\ \ell=m,\ i\neq \ell\\ 
      \frac{q}{s^2} & i=\ell,\ j=m,\ i\neq j\\
      \frac{q}{s^2} & i=m,\ j=\ell,\ i\neq j\\
      0 & \textrm{else}
    \end{cases}\\
    \intertext{and}
    (\E[Q^\top Q])_{i,j}(\E[Q^\top Q])_{\ell,m} &=\E[q_i^\top q_j]\E[q_\ell^\top q_m] \\
    & = \E\left[\sum_{k=1}^qq_{ik}q_{jk}\right]\E\left[\sum_{s=1}^q q_{\ell s} q_{m
      s}\right] = \frac{q^2}{s^2} I(i=j,\ \ell=m)\\
    \intertext{so,} 
    \V\left[ \vecfun{Q^\top Q}\right] &= \begin{cases}
      \frac{q}{s}-\frac{q}{s^2} & i=j=\ell=m\\
      \frac{q}{s^2} & i=\ell,\ j=m,\ i\neq j\\
      \frac{q}{s^2} & i=m,\ j=\ell,\ i\neq j\\
      0 & \textrm{else}.
    \end{cases}
  \end{align}
  Properties of the commutation matrix give the result.
\end{proof}

\begin{lemma}
\label{lem:coef-taylor-exp}
  For $A$ fixed, the first order Taylor expansion of the fully and
  partially compressed regression estimators are
  \begin{align}
  \hat\b_{FC}(\lambda) &= \hat\b_{ridge}(\lambda) + M(A-I_n) \hat{e} +
                         R_A, \textrm{ and}\\
 \hat\b_{PC}(\lambda) &= \hat\b_{ridge}(\lambda) - M(A-I_n)\X\hat\b_{ridge} + R_A,    
  \end{align}
where $\hat{e} = (I_n - H)Y$.
\end{lemma}

\begin{proof}[Proof of \autoref{lem:coef-taylor-exp}]
We begin with the fully compressed estimator.
\begin{align}
  \hat\b_{FC}(A) &= \hat\b_{FC}(I) + \frac{\partial
  \hat\b_{FC}(A)}{\partial \vecfun{A}}\Bigg|_{A=I} \vecfun{A-I} + R_A\\
                 &= \hat\b_{ridge} + \frac{\partial
                   (\X^\top A\X + \lambda I)^{-1}X^\top A y}{\partial \vecfun{A}}\Bigg|_{A=I} \vecfun{A-I} + R_A,
\end{align}
where $\vecfun{A} = \begin{bmatrix} a_{11} & a_{21} & \cdots &
  a_{nn}\end{bmatrix}^\top$.
Then,
\begin{align}
  \frac{\partial (\X^\top A\X + \lambda I)^{-1}X^\top A y}{\partial
  \vecfun{A}} &= \frac{\partial \vecfun{
                   (\X^\top A\X + \lambda I)^{-1}X^\top A y}}{\partial
             \vecfun{A}}\\
  &= \left[ \one \otimes (\X^\top A\X + \lambda
    I)^{-1} \right] \frac{\partial \vecfun{\X^\top A y}}{\partial
    \vecfun{A}} + \\
  &\quad +\left[y^\top A \X \otimes I\right] \frac{\partial \vecfun{(\X^\top
    A\X + \lambda I)^{-1}}} {\partial \vecfun{A}}.
\end{align}
Here we need two results:
\begin{align}
  \frac{\partial \vecfun{AFB}} {\partial \vecfun{F}} 
  &= B^\top \otimes A   \label{eq:harville1} \\
  \frac{\partial \vecfun{F^{-1}}} {\partial \vecfun{F}} 
  &= -(F^{-1})^\top \otimes F^{-1}.   \label{eq:harville2}
\end{align}
See \citep[][(16.6.8) and (16.6.15) respectively]{Harville1997}. Therefore,
\begin{align}
  \left[ \one \otimes (\X^\top A\X + \lambda
    I)^{-1} \right] \frac{\partial \vecfun{\X^\top A y}}{\partial
    \vecfun{A}} &= \left[ \one \otimes (\X^\top A\X + \lambda
    I)^{-1} \right] (y^\top \otimes \X^\top)\\
  &= y^\top \otimes (\X^\top A\X + \lambda I)^{-1} \X^\top,
\end{align}
and
\begin{align}
  \frac{\partial \vecfun{(\X^\top A\X + \lambda I)^{-1}}} {\partial \vecfun{A}} 
             &=   \frac{\partial \vecfun{(\X^\top
    A\X + \lambda I)^{-1}}} { \partial \vecfun{\X^\top
    A\X + \lambda I}} \frac{\partial\vecfun{  \X^\top
    A\X + \lambda I}} {\partial \vecfun{A}}\\
             &= -\left[ (\X^\top A\X + \lambda I)^{-1} \otimes (\X^\top
    A\X + \lambda I)^{-1} \right] \left[\X^\top \otimes \X^\top
               \right]\\
\intertext{so}
            \left[y^\top A \X \otimes I\right]  \frac{\partial
  \vecfun{(\X^\top A\X + \lambda I)^{-1}}} {\partial \vecfun{A}}  
             &= -\left[ y^\top A\X (\X^\top A \X + \lambda I)^{-1}
               \X^\top \right] \otimes \left[(\X^\top A \X + \lambda
               I)^{-1} \X^\top\right].
\end{align}
Therefore,
\begin{align}
  \frac{\partial \vecfun{(\X^\top A\X + \lambda I)^{-1} X^\top Ay}}
  {\partial \vecfun{A}}  
  &= \left[y^\top - y^\top A\X (\X^\top A \X + \lambda I)^{-1}
               \X^\top \right] \otimes \left[(\X^\top A \X + \lambda
               I)^{-1} \X^\top\right] \\
  &= (y-\X\hat\b_{FC})^\top \otimes \left[(\X^\top A \X + \lambda
               I)^{-1} \X^\top\right].
\end{align}
Finally,
\begin{align}
\hat\b_{FC}(A) &= \hat\b_{ridge} + (y-\X\hat\b_{FC})^\top \otimes
                 \left[(\X^\top A \X + \lambda 
                 I)^{-1} \X^\top\right]\Bigg|_{A=I}\vecfun{A-I} +
                 R_A\\
               &= \hat\b_{ridge} + (y-\X\hat\b_{ridge})^\top \otimes
                 \left[(\X^\top  \X + \lambda 
                 I)^{-1} \X^\top\right]\vecfun{A-I} + R_A\\
               &= \hat\b_{ridge} + \vecfun{ (\X^\top\X + \lambda
                 I)^{-1}\X^\top(A-I)(y-\X\hat\b_{ridge})} + R_A \label{eq:vecID}\\
  &= \hat\b_{ridge} + (\X^\top\X + \lambda
                 I)^{-1}\X^\top(A-I) (y-\X\hat\b_{ridge}) + R_A,
\end{align}
where the third equality follows from the result
$(B^\top\otimes A)\vecfun{X} = \vecfun{AXB}$.

The result for the partially compressed estimator proceeds similarly.
\begin{align}
  \hat\b_{PC}(A) &= \hat\b_{PC}(I) + \frac{\partial
                   \hat\b_{PC}(A)}{\partial \vecfun{A}}\Bigg|_{A=I} \vecfun{A-I} + R_A\\
                 &= \hat\b_{ridge} + \frac{\partial
                   (\X^\top A\X + \lambda I)^{-1}X^\top y}{\partial \vecfun{A}}\Bigg|_{A=I} \vecfun{A-I} + R_A.
\end{align}
Then, as above,
\begin{align}
  \frac{\partial (\X^\top A\X + \lambda I)^{-1}X^\top  y}{\partial
  \vecfun{A}} &= \frac{\partial \vecfun{
                   (\X^\top A\X + \lambda I)^{-1}X^\top y}}{\partial
             \vecfun{A}}\\
  &= \left[y^\top  \X \otimes I\right] \frac{\partial (\X^\top
    A\X + \lambda I)^{-1}} {\partial \vecfun{A}}\\
  &= -\left[ y^\top \X (\X^\top A \X + \lambda I)^{-1}
               \X^\top \right] \otimes \left[(\X^\top A \X + \lambda
               I)^{-1} \X^\top\right].
\end{align}
Finally,
\begin{align}
\hat\b_{PC}(A) &= \hat\b_{ridge} -\left[ y^\top \X (\X^\top A \X + \lambda I)^{-1}
               \X^\top \right] \otimes \left[(\X^\top A \X + \lambda
               I)^{-1} \X^\top\right]\Bigg|_{A=I}\vecfun{A-I} + R_A\\
               &= \hat\b_{ridge} -\left[ y^\top \X (\X^\top \X + \lambda I)^{-1}
                 \X^\top \right] \otimes \left[(\X^\top \X + \lambda
                 I)^{-1} \X^\top\right]\vecfun{A-I} + R_A\\
               &= \hat\b_{ridge} - \vecfun{(\X^\top \X + \lambda I)^{-1}
               \X^\top(A-I)\X(\X^\top \X + \lambda I)^{-1}\X^\top y} + R_A \\
  &= \hat\b_{ridge} - (\X^\top \X + \lambda I)^{-1}
               \X^\top(A-I)\X(\X^\top \X + \lambda I)^{-1}\X^\top y +
    R_A\\
  &= \hat\b_{ridge} - (\X^\top \X + \lambda I)^{-1}
               \X^\top(A-I)\X\hat\b_{ridge} + R_A.
\end{align}
\end{proof}

\begin{proof}[Proof of \autoref{thm:coefFC}]
  As $\E_Q[A] = I$, the conditional expectation is
  straightforward. For the conditional variance,
  \begin{align}
    \V_Q[\hat\b_{FC}\given Y] 
    &=(Y-\X\hat\b_{ridge})^\top \otimes
      \left[(\X^\top  \X + \lambda 
      I)^{-1} \X^\top\right] \V\left[ \vecfun{\frac{s}{q}
      Q^\top Q}\right] \times\\
    &\quad\times (Y-\X\hat\b_{ridge})\otimes \left[\X(\X^\top  \X + \lambda 
      I)^{-1} \right] + \V_Q[R_A]\\
    &= \left( (Y-\X\hat\b_{ridge})^\top \otimes
      \left[(\X^\top  \X + \lambda 
      I)^{-1} \X^\top\right]\right) \times \\
    &\quad\times \left(\frac{(s-3)_+}{q} \diag(\vecfun{I_n}) +
      \frac{1}{q} I_{n^2}+\frac{1}{q} K_{nn}\right) \times\\
    &\quad\times (Y-\X\hat\b_{ridge})\otimes \left[\X(\X^\top  \X + \lambda 
      I)^{-1} \right]+ \V_Q[R_A]\\
    &= \frac{(s-3)_+}{q} (\hat{e}^\top \otimes M )
      \diag(\vecfun{I_n}) (\hat{e} \otimes M^\top) +\\
    &\quad + \frac{1}{q}(\hat{e}^\top \otimes M ) (\hat{e} \otimes M^\top)
      + \frac{1}{q} (\hat{e}^\top \otimes M ) K_{nn} (\hat{e} \otimes M^\top) +\V_Q[R_A]\\
    &= \frac{(s-3)_+}{q} M\hat{e}\hat{e}^\top M^\top +
      \frac{1}{q}(\hat{e}^\top \hat{e})\otimes (M M^\top) 
      + \frac{1}{q} (M\otimes\hat{e}^\top)(\hat{e}\otimes M^\top) +\V_Q[R_A]\\
    &= \frac{(s-3)_+}{q} M\hat{e}\hat{e}^\top M^\top 
    + \frac{1}{q}\hat{e}^\top\hat{e} M M^\top + \frac{1}{q}
      M\hat{e}\hat{e}^\top M^\top +\V_Q[R_A].
  \end{align}
  For
  the unconditional case, taking the expectation of the conditional
  expectation with respect to $Y$ gives the first result. The second
  follows from the formula for total variance and
  \autoref{lem:ridge-residual-expect}. 
\end{proof}

\begin{proof}[Proof of \autoref{thm:coefPC}]
  Borrowing the notation from the previous result,
  \begin{align}
    \V_Q[\hat\b_{PC}\given Y] 
    &=(\hat Y^\top \otimes M)\V\left[ \vecfun{\frac{s}{q}
      Q^\top Q}\right] (\hat Y \otimes M^\top)+ \V_Q[R_A]\\
    &= \frac{(s-3)_+}{q}\left( \hat Y^\top \otimes
      M\right) \diag(\vecfun{I_n}) (\hat{Y} \otimes M^\top) +\\
    &\quad + \frac{1}{q}\left( \hat Y^\top \otimes
      M\right) (\hat{Y}\otimes M^\top) + \frac{1}{q}\left( \hat Y^\top \otimes
      M\right) K_{nn}(\hat{Y} \otimes M^\top) + \V_Q[R_A]\\
    &= \frac{(s-3)_+}{q}M\hat{Y}\hat{Y}^\top M^\top 
      + \frac{1}{q}
      \hat{Y}^\top\hat{Y} MM^\top + \frac{1}{q}M\hat{Y} \hat{Y}^\top M^\top  +\V_Q[R_A].
  \end{align}
  As before, taking the expectation of the conditional
  expectation with respect to $Y$ gives the first result. The second
  follows from the formula for total variance and
  \autoref{lem:ridge-hat-expect}. 
\end{proof}

The following two lemmas follow from standard
properties of expectation and variance of linear estimators, and their
proofs are omitted.

\begin{lemma}
  The expectation and variance of the fully compressed regression
  estimator, conditional on $Q$ and the design, are
  \begin{align*}
    \E[\hat\b_{FC}(\lambda)\given\X, Q]
    &= (\X^\top Q^\top Q\X + \lambda
                                  I)^{-1}\X^\top Q^\top Q \X\bstar,\\
\intertext{and}
    \V[\hat\b_{FC}(\lambda)\given\X, Q] 
    &=  \sigma^2(\X^\top Q^\top Q\X + \lambda
      I)^{-1}\X^\top (Q^\top Q)^2 \X (\X^\top Q^\top Q\X + \lambda
      I)^{-1}.
  \end{align*}

\end{lemma}

\begin{lemma}
  The expectation and variance of the partially compressed regression
  estimator, conditional on $Q$ and the design, are
  \begin{align*}
    \E[\hat\b_{PC}(\lambda)\given\X, Q]
    &= (\X^\top Q^\top Q\X + \lambda
                                  I)^{-1}\X^\top \X\bstar,\\
\intertext{and}
    \V[\hat\b_{PC}(\lambda)\given\X, Q] 
    &=  \sigma^2(\X^\top Q^\top Q\X + \lambda
      I)^{-1}\X^\top \X (\X^\top Q^\top Q\X + \lambda
      I)^{-1}.
  \end{align*}

\end{lemma}

We now use these lemmas to derive the following theorems for full and
partial compression.

\begin{proof}[Proof of \autoref{thm:bias2taylor}]
The squared bias for each estimator is
\begin{align}
  \bias^2(A) = \norm{\E[\hat\b(A)] - \bstar}_2^2.
\end{align}
Note that this is a function of $A$, and the expectation is over $y$
(and therefore $\epsilon$). We suppress the conditioning of all
expectations on $\X$ and $Q$. Taking a Taylor expansion to first order,
\begin{align}
  \bias^2(A) 
  &= \norm{\E[\hat\b(I)] - \bstar}_2^2 +
    \frac{\partial}{\partial \vecfun{A}} \norm{\E[\hat\b(A)] -
    \bstar}_2^2 \bigg|_{A=I} \vecfun{A-I} + R_A\\
  &= \norm{\E[\hat\b(I)] - \bstar}_2^2 + 2(\E[\hat\b(A)] -
    \bstar)^\top \frac{\partial}{\partial \vecfun{A}} \E[\hat\b(A)]
    \bigg|_{A=I} \vecfun{A-I} + R_A.
\end{align}
We have that
\begin{align}
  \E[\hat\b_{FC}(A)] &= (\X^\top A \X + \lambda I)^{-1} \X^\top A
                       \X\bstar\\
  \E[\hat\b_{PC}(A)] &= (\X^\top A \X + \lambda I)^{-1} \X^\top \X\bstar.
\end{align}
For the
fully pre-conditioned estimator, 
\begin{align}
  \frac{\partial (\X^\top A\X + \lambda I)^{-1}\X^\top A \X\bstar}{\partial
  \vecfun{A}} &= \frac{\partial \vecfun{
                   (\X^\top A\X + \lambda I)^{-1}\X^\top A \X\bstar}}{\partial
             \vecfun{A}}\\
  &= \left[ 1 \otimes (\X^\top A\X + \lambda
    I)^{-1}\right] \frac{\partial \vecfun{\X^\top A\X \bstar}}{\partial
    \vecfun{A}} + \\
  &\quad +\left[\bstar^\top\X^\top A \X \otimes I_p\right] \frac{\partial \vecfun{(\X^\top
    A\X + \lambda I)^{-1}}} {\partial \vecfun{A}}.
\end{align}
Proceeding as in the proof of \autoref{lem:coef-taylor-exp},
\begin{align}
\left[1\otimes  (\X^\top A\X + \lambda
    I)^{-1} \right] \frac{\partial \vecfun{\X^\top A \X\bstar}}{\partial
    \vecfun{A}} &= \left[ 1 \otimes (\X^\top A\X + \lambda
    I)^{-1} \right] (\bstar^\top\X^\top \otimes \X^\top)\\
  &= \bstar^\top\X^\top \otimes (\X^\top A\X + \lambda I)^{-1} \X^\top,
\end{align}
and
\begin{align}
  \frac{\partial \vecfun{(\X^\top A\X + \lambda I)^{-1}}} {\partial \vecfun{A}} 
             &=   \frac{\partial \vecfun{(\X^\top
    A\X + \lambda I)^{-1}}} { \partial \vecfun{\X^\top
    A\X + \lambda I}} \frac{\partial\vecfun{  \X^\top
    A\X + \lambda I}} {\partial \vecfun{A}}\\
             &= -\left[ (\X^\top A\X + \lambda I)^{-1} \otimes (\X^\top
    A\X + \lambda I)^{-1} \right] \left[\X^\top \otimes \X^\top
               \right]\\
  &= -\left[(\X^\top A\X + \lambda I)^{-1}\X^\top\right] \otimes \left[(\X^\top
    A\X + \lambda I)^{-1}\X^\top\right] \\
\intertext{so}
  \lefteqn{\left[\bstar^\top\X^\top A \X \otimes I_p\right]
  \frac{\partial\vecfun{ (\X^\top A\X + \lambda I)^{-1}}} {\partial \vecfun{A}}}\\
             &= -\left[ \bstar^\top\X^\top A\X (\X^\top A \X + \lambda
               I)^{-1} \X^\top \right] \otimes \left[(\X^\top A \X +
               \lambda I)^{-1} \X^\top\right].
\end{align}
Therefore,
\begin{align}
  \lefteqn{\frac{\partial \vecfun{(\X^\top A\X + \lambda I)^{-1} \X^\top A\X\bstar}}
  {\partial \vecfun{A}}}\\
  &= \left[\bstar^\top\X^\top - \bstar^\top\X^\top A\X (\X^\top A \X + \lambda I)^{-1}
               \X^\top \right] \otimes \left[(\X^\top A \X + \lambda
               I)^{-1} \X^\top\right].
\end{align}
Evaluating  
\begin{align*}
&2\left( (\X^\top A \X + \lambda I)^{-1} \X^\top A
                       \X\bstar-\bstar\right)\\ 
&\times \left[\bstar^\top\X^\top - \bstar^\top\X^\top A\X (\X^\top A \X + \lambda I)^{-1}
               \X^\top \right] \otimes \left[(\X^\top A \X + \lambda
               I)^{-1} \X^\top\right]
\end{align*}
at $A=I_n$ and simplifying gives
  \begin{align}
    &2(\E[\hat\b_{FC}(A)] -
    \bstar)^\top \frac{\partial}{\partial \vecfun{A}} \E[\hat\b_{FC}(A)]
    \bigg|_{A=I} \vecfun{A-I}\\
  &=2  \left(M
    \X\bstar - \bstar\right)^\top  \left[\bstar^\top\X^\top -
    \bstar^\top\X^\top H \right] \otimes M \vecfun{A-I_n}\\
    &= 2  (\left(M
    \X - I_p\right)\bstar)^\top  \left[\bstar^\top\X^\top (I_n -
      H) \right] \otimes M \vecfun{A-I_n}\label{eq:fcCondExpect}.
  \end{align}
We can simplify this expression using the following result from \citep[][16.2.15]{Harville1997}:
\[
\vecfun{A}^\top (D \otimes B) \vecfun{C} = \tr(A^\top B C D^\top).
\]
Therefore,
\begin{align}
  \lefteqn{2  (\left(M
  \X - I_p\right)\bstar)^\top  \left[\bstar^\top\X^\top (I_n -
  H) \right] \otimes M \vecfun{A-I_n}}\\
&= 2\tr\left(\bstar^\top\left[M\X
  - I_p\right] M (A-I_n) \left[I_n -
  H\right]\X\bstar\right)\\ 
&= 2\bstar^\top(M\X-I_p) M (A-I_n) (I_n - 
  H )\X\bstar.
\end{align}
Taking the expectation with respect to $Q$ gives the second result.

For the case of partial pre-conditioning,
\begin{align}
  \frac{\partial (\X^\top A\X + \lambda I)^{-1}\X^\top\X\bstar}{\partial
  \vecfun{A}} 
  &= \frac{\partial \vecfun{
    (\X^\top A\X + \lambda I)^{-1}\X^\top A \X\bstar}}{\partial
    \vecfun{A}}\\
  &= 
    \left[\bstar^\top\X^\top  \X \otimes I_p\right] \frac{\partial \vecfun{(\X^\top
    A\X + \lambda I)^{-1}}} {\partial \vecfun{A}}\\
  &= -\left[ \bstar^\top\X^\top \X (\X^\top A \X + \lambda
    I)^{-1} \X^\top \right] \otimes \left[(\X^\top A \X +
    \lambda I)^{-1} \X^\top\right]\\
  &= -\left[ \bstar^\top\X^\top \X (\X^\top A \X + \lambda
    I)^{-1} \X^\top \right] \otimes \left[(\X^\top A \X +
    \lambda I)^{-1} \X^\top\right]
\end{align}
Proceeding as above and evaluating at $A=I_n$,
\begin{align}
  &2(\E[\hat\b_{PC}(A)] -
    \bstar)^\top \frac{\partial}{\partial \vecfun{A}} \E[\hat\b_{PC}(A)]
    \bigg|_{A=I} \vecfun{A-I}\\
  &=-2  (\left(M
    \X - I_p\right)\bstar)^\top  \left[\bstar^\top\X^\top H
    \right] \otimes M \vecfun{A-I_n}\label{eq:pcCondExpect}\\
&= -2\tr\left(\bstar^\top\left[M\X
  - I_p\right] M (A-I_n)
  H\X\bstar\right)\\ 
&= 2\bstar^\top\left[I_p-M\X \right] M (A-I_n)
  H\X\bstar.
\end{align}

\end{proof}

\begin{proof}[Proof of \autoref{thm:trace-variance-taylor}]
  As for the trace of the variance, the Taylor expansion is given by
\[
\tr(\V[\hat\b\given \X, Q]) = \tr(\V[\hat\b_{ridge}\given \X]) +
\frac{\partial}{\partial \vecfun{A}} \tr(\V[\hat\b\given
\X, Q])\Bigg|_{A=I}\vecfun{A-I_n} + R_A.
\]

For any matrix $B$, we have
$\tr(B^\top B) = \vecfun{B}^\top\vecfun{B}$. Therefore,
\begin{align}
  \tr(\V[\hat\b_{FC}(A) \given \X, Q]) &= \tr((\X^\top A \X + \lambda
  I)^{-1}\X^\top A^\top A \X (\X^\top A \X + \lambda
  I)^{-1})\\
  \tr(\V[\hat\b_{PC}(A) \given\X, Q]) &= \tr((\X^\top A \X + \lambda
  I)^{-1}\X^\top \X (\X^\top A \X + \lambda
  I)^{-1}),
\end{align}
which both have this form. We begin with fully compressed.
\begin{align}
  &\frac{\partial}{\partial\vecfun{A}} \tr((\X^\top A \X + \lambda
  I)^{-1}\X^\top A^\top A \X (\X^\top A \X + \lambda
  I)^{-1})\\
  &= 2\vecfun{A\X(\X^\top A \X + \lambda
  I)^{-1}}^\top \frac{\partial}{\partial \vecfun{A}} \vecfun{A\X(\X^\top A \X + \lambda
  I)^{-1}}.
\end{align}
Now, we have
\begin{align}
  &\frac{\partial}{\partial \vecfun{A}} \vecfun{A\X(\X^\top A \X + \lambda
  I)^{-1}}  \\
  &= [I_p \otimes A\X] \frac{\partial}{\partial \vecfun{A}} \vecfun{(\X^\top A \X + \lambda
  I)^{-1}}+\\
  &\quad+ [(\X^\top A \X + \lambda I)^{-1} \otimes I_n]
    \frac{\partial}{\partial \vecfun{A}} \vecfun{A\X},
\end{align}
where
\begin{align}
  \frac{\partial}{\partial \vecfun{A}} \vecfun{(\X^\top A \X + \lambda
  I)^{-1}} 
  &=  -\left[(\X^\top A\X + \lambda I)^{-1}\X^\top\right] \otimes \left[(\X^\top
    A\X + \lambda I)^{-1}\X^\top\right]\\
\intertext{and}
  \frac{\partial}{\partial \vecfun{A}} \vecfun{A\X}
  &= \X^\top \otimes I_n.
\end{align}
Combining these gives
\begin{align}
  &\frac{\partial}{\partial \vecfun{A}} \vecfun{A\X(\X^\top A \X + \lambda
    I)^{-1}}  \\
  &=- [I_p \otimes A\X] \left[(\X^\top A\X + \lambda I)^{-1}\X^\top\right] \otimes \left[(\X^\top
    A\X + \lambda I)^{-1}\X^\top\right] +\\
  &\quad + [(\X^\top A \X + \lambda I)^{-1} \otimes I_n][\X^\top
    \otimes I_n]\\
  &= -(\X^\top A\X + \lambda I)^{-1}\X^\top \otimes A\X (\X^\top A\X +
    \lambda I)^{-1}\X^\top +\\
  &\quad + (\X^\top A\X + \lambda I)^{-1}\X^\top \otimes I_n\\
  &= -(\X^\top A\X + \lambda I)^{-1}\X^\top \otimes [A\X (\X^\top A\X +
    \lambda I)^{-1}\X^\top - I_n].
\end{align}
Evaluating at $A=I_n$ and simplifying gives
  \begin{align}
    &\frac{\partial}{\partial \vecfun{A}} \tr(\V[\hat\b_{FC}\given
    \X,Q])\Bigg|_{A=I}\vecfun{A-I_n}\\
    &=-2\vecfun{M^\top}^\top \left[M \otimes (H- I_n)\right] \vecfun{A - I_n}\\
    &= 2 \tr\left(M(I_n-H)(A-I_n)M^\top\right).
  \end{align}

For the case of partial compression, we only need 
\begin{align}
  &\frac{\partial}{\partial \vecfun{A}} \vecfun{\X(\X^\top A \X + \lambda
    I)^{-1}} \\
  &= [I_p \otimes \X] \frac{\partial}{\partial \vecfun{A}} \vecfun{(\X^\top A \X + \lambda
  I)^{-1}}\\
  &= - [I_p \otimes \X] \left[(\X^\top A\X + \lambda I)^{-1}\X^\top\right] \otimes \left[(\X^\top
    A\X + \lambda I)^{-1}\X^\top\right]\\
  &= - [(\X^\top A\X + \lambda I)^{-1}\X^\top] \otimes [\X (\X^\top A\X + \lambda I)^{-1}\X^\top]
\end{align}
Evaluating at $A=I_n$ and simplifying gives
  \begin{align}
    \frac{\partial}{\partial \vecfun{A}} \tr(\V[\hat\b_{FC}\given
    \X,Q])\Bigg|_{A=I}\vecfun{A-I_n}
    &=-2\vecfun{M^\top}^\top M \otimes H \vecfun{A - I_n}\\
    &= -2 \tr\left(MH(A-I_n)M^\top\right).
  \end{align}

For the variances unconditional on $Q$, note that by the law of total
variance,
\begin{align*}
  \tr(\V[\hat\b \given \X]) &= \tr\left(\E[\V[\hat\b\given \X, Q]]\right)
+\tr\left(\V[\E[\hat\b\given \X, Q]]\right)\\
  &= \E[\tr\left(\V[\hat\b\given \X, Q]\right)]
+\tr\left(\V[\E[\hat\b\given \X, Q]]\right)\\
  &= \sigma^2 \sum_{j=1}^p \frac{d_j^2} {(d_j^2 + \lambda)^2} +\E[R_A]
    + \tr\left(\V[\E[\hat\b\given \X, Q]]\right)
\end{align*}
By the conditional part of this theorem. Now, we first expand
$\E[\hat\b\given \X, Q]$, relying on previous results:
\begin{align*}
  \E[\hat\b\given \X, Q] 
  &= (\X^\top \X + \lambda I)^{-1}\X^\top\X\b_*
    + \frac{\partial}{\partial A}
    \E[\hat\b\given \X, Q] \Bigg|_{A=I}\vecfun{A-I_n} + R_A.
\end{align*}
As we will be taking the variance with respect to $A$, the first term
is constant. For the second term we have,
\begin{align*}
  \E[\hat\b_{FC}\given \X, Q] \Bigg|_{A=I}\vecfun{A-I_n} 
  &= [\b^\top_*\X^\top(I_n-H)]\otimes M \vecfun{A-I_n}  & \textrm{by
                                                          \eqref{eq:fcCondExpect}}\\
  &=\vecfun{M(A-I_n)(I_n-H)\X\b_*}\\
  &=M(A-I_n)(I_n-H)\X\b_*\\
  \E[\hat\b_{PC}\given \X, Q] \Bigg|_{A=I}\vecfun{A-I_n} 
  &=  \left[\bstar^\top\X^\top H
    \right] \otimes M \vecfun{A-I_n}&\textrm{by
                                      \eqref{eq:pcCondExpect}}\\
  &= \vecfun{M(A-I_n)H\X\b_*}\\
  &= M(A-I_n)H\X\b_*.
\end{align*}
Therefore,
\begin{align*}
   \tr(\V[\hat\b_{FC} \given \X]) &=\sigma^2 \sum_{i=1}^p \frac{d_i^2}
                                    {(d_i^2 + \lambda)^2} +\E[R_A] +
                                    \tr(\V[MA(I_n-H)\X\b_*]) +
                                    \tr(\V[R_A])\\
   \tr(\V[\hat\b_{PC} \given \X]) &=\sigma^2 \sum_{i=1}^p \frac{d_i^2}
                                    {(d_i^2 + \lambda)^2} +\E[R_A] +
                                    \tr(\V[MAH\X\b_*]) + \tr(\V[R_A]).
\end{align*}
In both cases, we need $\tr(\V[LAz])$ for a matrix $L$ and vector
$z$. As this is a vector, we have that $LAz = (z^\top\otimes
L)\vecfun{A}$, so 
\begin{align*}
  \V[LAz] 
  &=(z^\top\otimes
    L)\V[\vecfun{A}] (z\otimes
    L^\top)\\
  &=(z^\top\otimes
L)\left[\frac{(s-3)_+}{q} \diag(\vecfun{I_n}) +
  \frac{1}{q} I_{n^2}+\frac{1}{q} K_{nn}\right] (z \otimes
    L^\top)\\
  &=\frac{(s-3)_+}{q} (z^\top\otimes
    L) \diag(\vecfun{I_n}) (z\otimes
    L^\top)+\\
  &\quad+\frac{1}{q}(z^\top\otimes
    L) (z \otimes L^\top) +
    \frac{1}{q}(z^\top\otimes
    L)K_{nn} (z\otimes L^\top) \\
  &=\frac{(s-3)_+}{q} Lzz^\top L^\top+\frac{1}{q}(z^\top z\otimes LL^\top) +
    \frac{1}{q}(L\otimes z^\top)(z \otimes L^\top) \\
  &=\frac{(s-3)_+}{q} Lzz^\top L^\top+\frac{z^\top z}{q}LL^\top +
    \frac{1}{q}Lzz^\top L^\top\\
  &=\frac{(s-2)_+}{q} Lzz^\top L^\top+\frac{z^\top z}{q}LL^\top.
\end{align*}
Now applying the trace,
\begin{align*}
  \tr(\V[LAz] ) 
  &= \frac{(s-2)_+}{q} \tr(Lzz^\top L^\top)+\frac{z^\top z}{q}\tr(LL^\top).
\end{align*}
Therefore,
\begin{align*}
  \tr(\V[MA(I_n-H)\X\b_*]) &=
    \frac{(s-2)_+}{q}\tr\left(M(I_n-H)\X\bstar\bstar^\top\X^\top(I_n-H)M^\top\right)\\
  &\quad + \frac{1}{q}\bstar^\top\X^\top(I-H)^2\X\bstar\tr(M^\top M)\\ 
  \tr(\V[MAH\X\b_*]) &=
    \frac{(s-2)_+}{q}\tr\left(MH\X\bstar\bstar^\top\X^\top HM^\top\right)\\
  & \quad+ \frac{1}{q}\bstar^\top\X^\top H^2\X\bstar\tr(M^\top M).
\end{align*}
Using the SVD of $\X$ gives,
\begin{align*}
  M(I_n-H)\X 
  &= V(D^2+\lambda I_p)^{-1}DU^\top(I_n - UD(D^2+\lambda
    I_p)^{-1}DU^\top) UDV^\top\\
  &= V(D^2+\lambda I_p)^{-1} D(I_p - D(D^2+\lambda
    I_p)^{-1}D)DV^\top\\
  &= \lambda V(D^2+\lambda I_p)^{-1}D^2(D^2+\lambda I_p)^{-1}V^\top\\
  &= \lambda VD^2(D^2+\lambda I_p)^{-2}V^\top\\
  \X^\top (I_n-H)^2\X
  &= VDU^\top\left( \lambda^2 U(D^2+\lambda I_p)^{-2}U^\top\right) UDV^\top\\
  &= \lambda^2VD^2(D^2+\lambda I_p)^{-2}V^\top\\
  MH\X 
  &= V(D^2+\lambda I_p)^{-1}DU^\top U D(D^2+\lambda I_p)^{-1}DU^\top U D V^\top\\
  &= V(D^2+\lambda I_p)^{-1}D^2(D^2+\lambda I_p)^{-1}D^2 V^\top\\
  &= VD^4(D^2+\lambda I_p)^{-2}V^\top\\
  \X^\top H^2\X
  &= \X^\top \X(\X^\top X + \lambda I_p)^{-1}\X^\top\X(\X^\top X +
    \lambda I_p)^{-1}\X^\top\X\\
  &= VD^6(D^2+\lambda I_p)^{-2}V^\top.
\end{align*}
Therefore,
\begin{align*}
  \tr(\V[MA(I_n-H)\X\b_*]) &=
    \frac{\lambda^2(s-2)_+}{q}\bstar^\top V D^2(D^2+\lambda
                             I)^{-4}D^2V^\top \bstar\\
  &\quad + \frac{\lambda^2}{q}\bstar^\top VD(D^2+\lambda I)^{-2}DV^\top\bstar\sum_{j=1}^p \frac{d_j^2} {(d_j^2 + \lambda)^2}\\ 
  \tr(\V[MAH\X\b_*]) &=
    \frac{(s-2)_+}{q}\bstar^\top VD^2(D^2+\lambda I)^{-2}D^2 V^\top\bstar\\
  & \quad+ \frac{1}{q}\bstar^\top VD^3(D^2+\lambda I)^{-2}D^3 V^\top\bstar\sum_{j=1}^p \frac{d_j^2} {(d_j^2 + \lambda)^2}.
\end{align*}

\end{proof}

\begin{proof}[Proof of \autoref{cor:lambda-star}]
  For full compression we have,
  \begin{align*}
    \tr\left(\frac{\lambda^2(s-2)_+}{q}\bstar^\top V D^2(D^2+\lambda
                             I)^{-4}D^2V^\top \bstar\right) 
    &= \frac{\lambda^2(s-2)_+}{q}\frac{b^2p n^2}{(n+\lambda)^4} \\
    \tr\left(\frac{\lambda^2}{q}\bstar^\top VD(D^2+\lambda
    I)^{-2}DV^\top\bstar\right)\sum_{j=1}^p \frac{d_j^2} {(d_j^2 + \lambda)^2} 
    &= \frac{\lambda^2}{q}\frac{b^2pn}{(n+\lambda)^2}\frac{pn}{(n+\lambda)^2}.
  \end{align*}

  For partial compression we have,
  \begin{align*}
    \tr\left(\frac{(s-2)_+}{q}\bstar^\top VD^2(D^2+\lambda I)^{-2}D^2 V^\top\bstar\right) 
    &= \frac{(s-2)_+}{q}\frac{b^2 p n^2}{(n+\lambda)^2} \\
    \tr\left(\frac{1}{q}\bstar^\top VD^3(D^2+\lambda I)^{-2}D^3 V^\top\bstar\right)\sum_{j=1}^p \frac{d_j^2} {(d_j^2 + \lambda)^2} 
    &= \frac{\lambda^2}{q}\frac{b^2pn^3}{(n+\lambda)^2}\frac{pn}{(n+\lambda)^2}.
  \end{align*}
  Simplifying and substituting $\theta=\lambda/n$ gives the result.
\end{proof}
\end{document}